\documentclass[11pt]{article}
\usepackage{hyperref}
\usepackage[utf8]{inputenc}
\usepackage[T1]{fontenc}
\usepackage{amsmath}  
\usepackage{fullpage} %
\usepackage{amsfonts,amsthm}  
\usepackage{natbib}
\usepackage{cleveref}
\usepackage{algorithm}
\usepackage{graphicx}
\usepackage{enumerate}
\usepackage[noend]{algpseudocode} 
\usepackage{url}
\usepackage{subcaption}

\usepackage{thm-restate} 
\usepackage{natbib}
\setcitestyle{numbers,square,sorted}
\algrenewcommand\algorithmicrequire{\textbf{Input:}}
\algrenewcommand\algorithmicensure{\textbf{Output:}}

\newtheorem{lemma}{Lemma}

\newtheorem{definition}{Definition}

\usepackage{xcolor}
\newcounter{note}[section]

\newcommand{\dist}{\textrm{dist}}

\newcommand{\be}{\begin{enumerate}}
\newcommand{\ee}{\end{enumerate}}

\newcommand{\bi}{\begin{itemize}}
\newcommand{\ei}{\end{itemize}}

\newcommand{\mc}[1]{\mathcal{#1}}
\newcommand{\mb}[1]{\mathbb{#1}}

\newcommand{\E}{\mathbb E}

\newcommand{\poly}{\text{poly}}

\newcommand{\lmt}{\left[\begin{smallmatrix}}
\newcommand{\rmt}{\end{smallmatrix}\right]}

\newcommand{\shortOrLong}{full}

\ifthenelse{\equal{\shortOrLong}{short}}{
  \wlog{SHORT VERSION}
  
  \newcommand{\fullOnly}[1]{}
  \newcommand{\shortOnly}[1]{#1}
}{
  \wlog{FULL VERSION}
  \newcommand{\fullOnly}[1]{#1}
  \newcommand{\shortOnly}[1]{}
}

\begin{document}

\newcommand*\samethanks[1][\value{footnote}]{\footnotemark[#1]}
\begin{titlepage}

\title{Erasure Correction for Noisy Radio Networks}
\author{Keren Censor-Hillel\thanks{Supported in part by the Israel Science Foundation (grant 1696/14) and the Binational Science Foundation (grant 2015803).}, Bernhard Haeupler\thanks{Supported in part by NSF grants CCF-1527110, CCF-1618280 and NSF CAREER award CCF-1750808.}, D Ellis Hershkowitz\samethanks, Goran Zuzic\samethanks}
\date{\today}

\maketitle

\begin{abstract}
  The radio network model is a well-studied model of wireless, multi-hop networks. However, radio networks make the strong assumption that messages are delivered deterministically. The recently introduced noisy radio network model relaxes this assumption by dropping messages independently at random.
  
  In this work we quantify the relative computational power of noisy radio networks and classic radio networks. In particular, given a non-adaptive protocol for a fixed radio network we show how to reliably simulate this protocol if noise is introduced with a multiplicative cost of $\mathrm{poly}(\log \Delta, \log \log n)$ rounds where $n$ is the number nodes in the network and $\Delta$ is the max degree. Moreover, we demonstrate that, even if the simulated protocol is not non-adaptive, it can be simulated with a multiplicative $O(\Delta \log ^2 \Delta)$ cost in the number of rounds. Lastly, we argue that simulations with a multiplicative overhead of $o(\log \Delta)$ are unlikely to exist by proving that an $\Omega(\log \Delta)$ multiplicative round overhead is necessary under certain natural assumptions.
\end{abstract}

\vfill
\thispagestyle{empty}
\end{titlepage}

\section{Introduction}\label{sec:intro}
The study of distributed graph algorithms provides precise mathematical models to understand how to accomplish distributed tasks with minimal communication. A classic example of such a model is the radio network model of \citet{onBroadChlamtac}. Radio networks use synchronous rounds of communication and were designed to model the collisions that occur in multi-hop, wireless networks. However, radio networks make the strong assumption that, provided no collisions occur, messages are guaranteed to be delivered. 

This assumption is overly optimistic for real environments in which noise may impede communication, and so our previous work~\cite{haeupler2017broad} introduced the noisy radio network model where messages are randomly dropped with a constant probability. This prior work demonstrated that the runtime of existing state-of-the-art broadcast protocols deteriorates significantly in the face of random noise. Furthermore, it showed how to design efficient broadcasting protocols that are robust to noise.

However, it has remained unclear how much more computationally powerful radio networks are than noisy radio networks. In particular, it was not known how efficiently an arbitrary protocol from a radio network can be \emph{simulated} by a protocol in a noisy radio network. A simulation with little overhead would demonstrate that radio networks and noisy radio networks are of similar computational power. However, if simulation was necessarily expensive then radio networks would be capable of completing more communication-intensive tasks than their noisy counterparts.

A simple observation regarding the relative power of these two models is that any polynomial-length radio network protocol can be simulated at a multiplicative cost of $O(\log n)$ rounds where $n$ is the number of nodes in the network. In particular, we can simulate any protocol from the non-noisy setting by repeating every round $O(\log n)$ times. A standard Chernoff and union bound argument show that every message sent by the original protocol is successfully sent with high probability. However, a multiplicative $O(\log n)$ is a significant price to pay for noise-robustness for many radio network protocols. For example, the optimal known-topology message broadcast protocol of \citet{gkasieniec2007faster} uses only $O(D + \log ^2 n)$ rounds, while the topology-oblivious Decay protocol of~\citet{bar1992time} takes $O(D \log n + \log^2 n)$ where $D$ is the diameter of the network. Moreover, a dependence on a global property of the network---the number of nodes in the graph---is excessive for correcting for a local issue---faults. 

The simple solution from above \emph{globally synchronizes} nodes by forcing all nodes to simulate the same round for $O(\log n)$ repetitions. A natural question is can one more efficiently simulate a noisy protocol by enforcing only \emph{local synchronization}. In this paper we show how to use local synchronization to efficiently simulate radio network protocols in a noisy setting.

The local synchronization technique we use is as follows. Suppose each node tracked the round in the original protocol up to which it has successfully simulated; we call this round a node's \emph{virtual round}. In our simulation in each round each node simulate the virtual round of its neighbor which has successfully simulated the fewest total rounds; i.e.\ each node simulates the virtual round of its neighbor with the largest ``delay'', thereby ``helping'' it. Such a simulation is local as---unlike the above simple simulation---nodes in distant parts of the network may simulate different rounds of the original protocol. Moreover, if a node has successfully simulated fewer rounds than all of its neighbors (and all of its neighbors' neighbors) then after one more round of simulation it will successfully simulate a new round of the original protocol if no random faults occur.

However, there are at least three notable challenges in implementing and proving the efficiency of such a local-synchronization-based simulation.
\begin{enumerate}
\item First, one must show that locally synchronizing nodes yields a fast simulation. A priori, it is not clear that locally synchronizing nodes provides an advantage over the simple global synchronization strategy.
	
\item Second, one must deal with the fact that local synchronization requires nodes to determine the minimal virtual rounds of its neighbors. In particular, nodes cannot easily compute the virtual rounds of their neighbors without communicating: When node $v$ simulates a round of the original protocol and broadcasts should it assume all of its neighbors have now successfully simulated this round? If a random fault occurred at a receiver then clearly $v$ should not but $v$ has no easy means of determining whether or not any such faults occurred. One must, therefore, determine how nodes can efficiently determine the minimal virtual round of their neighbors.
	
\item Lastly, one must overcome the fact that not receiving a message is indistinguishable from a randomly dropped message. In particular, a node can interpret not receiving a message in a simulated round as indicating either (1) that it receives no message in the simulated round of the original protocol or (2) that it does receive a message in this simulated round but a random fault dropped this message. Thus, it is not clear how nodes ought to advance their virtual rounds in the absence of messages. On the one hand, failing to increment a node's virtual round in the former case could needlessly slow down the simulation. On the other hand, if a node incorrectly decides it does not receive a message in a round, its idea of what messages it received will diverge from that of its neighbor who tried to send it a message. This divergence, in turn, may compound into further errors later in the simulation.

\end{enumerate}

\subsection{Our Contributions}

In this work we  present solutions for these challenges which demonstrate that, by using local synchronization, radio network protocols can be simulated by noisy radio networks with a multiplicative dependence on $\Delta$, the maximum degree of the network. Thus, we demonstrate that, for low-degree networks, noisy radio networks are essentially of the same computational power as radio networks. Moreover, we demonstrate that this dependence is more or less tight in two natural settings.

We give our simulations in three increasingly difficult settings, each of which introduces one of the above three challenges. In particular, our first model focuses on the first challenge, our second focuses on the first two challenges and our last model deals with all three challenges. 
We briefly mention our techniques here but defer more thorough intuition regarding our simulation techniques until \Cref{sec:intution}, which is after we have more formally defined our problem. 

\begin{enumerate}
\item \textbf{Local Progress Detection } As a warmup we begin by providing a simulation with $O(\log \Delta)$ multiplicative round overhead in the setting where nodes have access to ``\emph{local progress detection}''. Roughly, local progress detection enables nodes to know when they experience a random fault and also the virtual rounds of their neighbors. Notice that in this setting challenges 2 and 3 are non-issues: if each node has local progress detection, they can easily determine neighbors' delays and distinguish not receiving a message in the original protocol from a randomly dropped message. The key technique we introduce for the first challenge is a concentration-inequality-type result based on what we call ``blaming chains''. (\Cref{sec:warmup-progress-detection})


\item \textbf{Non-Adaptive Protocols } Next, we provide simulations for \emph{non-adaptive} protocols.  Roughly, non-adaptive protocols are protocols in which nodes know a priori the rounds in which they receive messages. In this setting, our simulations achieve a multiplicative $ \poly(\log \Delta, \log \log n) = O(\log^3 \Delta \cdot \log \log n \cdot \log \log \log n)$ round overhead \emph{without local progress detection}. As we argue in \Cref{sec:static-simulators}, a number of well-known protocols (e.g.\ the optimal broadcast algorithm of \citet{gkasieniec2007faster}) are non-adaptive. Notice that in this setting we need not deal with challenge 3 above since nodes can always distinguish a dropped message from a round in which they receive no messages because they know the rounds of the original protocol in which they receive messages. In this setting, we leverage non-adaptiveness of protocols to overcome challenge 2. In particular, we use a distributed binary search which carefully silences nodes that search over divergent ranges to inform nodes of the minimal virtual round of their neighbors. Moreover, we keep the range over which the binary search must search tractable by slowing down nodes that make progress too quickly.  (\Cref{sec:static-simulators})


\item \textbf{General Protocols } Lastly, we show how to deal with all three challenges at once by giving simulations for arbitrary radio network protocols with a multiplicative $O(\Delta \log^2 \Delta)$ overhead. To overcome challenge 3 we have nodes exchange ``tokens'' with all neighbors in every round to preclude the possibility of a dropped message going unnoticed. (\Cref{sec:general-simulator}) 

\item \textbf{Lower Bounds } We also show that our simulations for non-adaptive protocols are likely optimal: We show that two natural classes of simulations necessarily use $\Omega(\log \Delta)$ multiplicatively many more rounds than the protocols which they simulate. In particular, a simulation that either (1) does not use network coding, or (2) sends information in the same way as the original protocol requires $\Omega(\log \Delta)$ multiplicatively many more rounds than the original protocol. We also give a construction which we believe gives an unconditional $\Omega(\log \Delta)$ simulation overhead. (\Cref{sec:lower-bounds})
\end{enumerate}

\section{Related Work}
Let us review some related work.
\subsection{Robust Communication}
Several models have been studied to understand robust communication in models similar to the radio network model. \citet{ElGamal1984} introduced the noisy broadcast model, which also assumes random errors, and focuses on studying the computation of functions of the inputs of nodes~\cite{Gallager88, Newman04, GoyalKS08, kushilevitz1998computation}. The model of \citet{ElGamal1984} differs from our model in that it assumes a complete communication network and single-bit transmissions. \citet{rajagopalan1994coding} introduced a similar model which again assumes single-bit transmissions and also does not have collisions as our model does. Several papers have been written on notions of noisy radio networks which assume the network admits geometric structure. \citet{kranakis1998fault} considers broadcasting in radio networks where unknown nodes fail. In this work, unlike our own, nodes in the network admit some geometric structure; e.g.\ every node is at integer points on the line. In a similar vein, \citet{kranakis2008communication} consider radio networks with possibly correlated faults at nodes. Like the previous work, this work assumes that there is an underlying geometric structure to the radio network; namely the network is a disc graph. There have also been several papers on noisy single-hop radio networks. See either \citet{gilbert2009interference} for a study of an adversarial model or \citet{efremenko2018interactive} for a nice study of a random noise model. Another model which captures uncertainty is the \emph{dual graph} model~\cite{kuhn2009abstract,Censor-HillelGKLN14,KuhnLNOR10,GhaffariHLN12,GhaffariLN13}, in which an adversary chooses a set of unreliable edges in each round. 

There has also been extensive work on two-party interactive communication in the presence of noise \cite{schulman1996coding,haeupler2014interactive,gelles2017coding}. In this setting, Alice and Bob have some conversation in mind they would like to execute over a noisy channel; by adding redundancy they hope to hold a slightly longer conversation from which they can recover the original conversation outcome even when a fraction of the coded conversation is corrupted. Multi-party generalizations of this problem have also been studied \citet{braverman2017constant}. The noisy radio network model can be seen as a radio network analogue of these interactive communication models in which erasures occur rather than corruptions.

Lastly, work on MAC layers has sought to provide abstractions for algorithms that hide low-level uncertainty in wireless communication \cite{ghaffari2014multi,kuhn2009abstract,richa2008jamming,khabbazian2014decomposing}. Radio networks differ from MAC layers as in radio networks it is not required that a sender receive an acknowledgment from a receiver.

\subsection{Radio Networks}
Since its introduction by~\citet{onBroadChlamtac}, the classic radio network model has attracted wide attention from researchers. The survey of~\cite{Peleg2007} is an excellent overview of this research area. Here we focus the radio network literature that relates to our work. Much of previous work for the noisy radio network model focused on broadcast \cite{kowalski2005time}. For the classic model, ~\citet{bar1992time} gave a single-message broadcast algorithm for a known topology, which completes in $O(D\log n + \log^2 n)$ rounds. \citet{haeupler2017broad} shows that this protocol is robust to noise, completing in $O(\frac{\log{n}}{1-p}(D+\log{n}+\log{\frac{1}{\delta}}))$ rounds, with a probability of failure of at most $\delta$.
In the classic model, single-message broadcast was then improved by~\citet{gkasieniec2007faster} and ~\citet{Kowalski2007}, who showed that in the case of a known topology, $O(D + \log^2 n)$ rounds suffice. \citet{haeupler2017broad} shows that this protocol is \emph{not} robust to noise, requiring in expectation $\Theta(\frac{p}{1-p}D\log{n}+\frac{1}{1-p}D)$ rounds, for broadcasting a message along a path of length $D$. They showed an alternative protocol, completing in $O(D+\log{n}\log\log{n}(\log{n}+\log{\frac{1}{\delta}}))$ rounds, with a probability of failure of at most $\delta$.
For an unknown topology in the classic model, \citet{Czumaj2006115} give a protocol completing in $O(D \log (n/D) + \log^2 n)$ rounds, which is optimal, due to the $\Omega(\log^2{n})$ and $\Omega(D \log (n/D))$ lower bounds of~\citet{alon1991lower} and~\citet{Kushilevitz1993}, respectively. \citet{ghaffari2015randomized} give a $O(D + \poly\log n)$-round protocol that uses \emph{collision detection}.

\subsection{Simulations}
Lastly, simulations of models of distributed computation by other models of distributed computation is a foundational aspect of distributing computing, dating back to the 80's--90's with many simulations of various shared memory primitives, faults, and more (see a wide variety in, e.g.,~\citet{AttiyaWBook}). It is also a focus for message-passing models with different features, being the motivation for synchronizers~\citep{Awerbuch1985,AwerbuchP1990}, and additional simulations~\citep{Censor-HillelHK17,Censor-HillelGH18}.

\section{Model and Assumptions}
In this section we formally define the classic radio network model, the noisy radio network model and discuss various assumptions we make throughout the paper.

\textbf{Radio Networks } A multi-hop radio network, as introduced in \citet{onBroadChlamtac}, consists of an undirected graph $G = (V, E)$ with $n := |V|$ nodes. Communication occurs in synchronous rounds: in each round, each node either broadcasts a single message containing $\Theta(\log n)$ bits to its neighbors or listens. A node receives a message in a round if and only if it is listening and \emph{exactly} one of its neighbors is transmitting a message. If two or more neighbors of node $v$ transmit in a single round, their transmissions are said to \emph{collide} at $v$ and $v$ does not receive anything. We assume no collision detection, meaning a node cannot differentiate between when none of its neighbors transmit a message and when two or more of its neighbors transmit a message. Nodes are also typically assumed to have unbounded computation, though all of the protocols in our paper use polynomial computation.

We use the following notational conventions throughout this paper when referring to radio networks. Let $\mathrm{dist}(v, w)$ be the hop-distance between $u$ and $v$ in $G$, let $\Gamma(v) = \{ w \in V : \textrm{dist}(v, w) \le 1 \}$ denote the 1-hop neighborhood of $v$, and let $\Gamma^{(k)}(v) = \{ w \in V : \textrm{dist}(v, w) \le k \}$ denote the $k$-hop neighborhood of $v$.

\textbf{Noisy Radio Networks } A multi-hop noisy radio network, as introduced in \citet{haeupler2017broad}, is a radio network with random erasures. In particular, it is a radio network where node $v$ receives a message in a round if and only if it is listening, exactly one of its neighbors is broadcasting and a \emph{receiver fault} does not occur at $v$. Receiver faults occur at each node and in each round independently with constant probability $p \in (0, 1)$. As $p$ will be treated as a constant it will be suppressed in our $O$ and $\Omega$ notation. We assume that in a given round a node cannot differentiate between a message being dropped because of a fault, i.e., a collision, and all of its neighbors remaining silent.

We consider this model because independent receiver faults model transient environmental interference such as the capture effect \cite{leentvaar1976capture}.\footnote{In contrast, a sender faults model in which an entire broadcast by a node is dropped might model hardware failures where independence of faults would be a poor modeling choice.} Moreover, we consider an \emph{erasure} model---entire messages are dropped---rather than a corruption model---messages are corrupted at the bit level--- because, in practice, wireless communication typically incorporates error correction and checksums that can guard against bit corruptions \cite{elaoud1998adaptive}. The noisy radio network model can also be seen as modeling those cases when this error correction fails and the message cannot be reconstructed and is therefore effectively dropped.

\textbf{Protocols } A protocol governs the broadcast and listening behavior of nodes in a network. In particular, a protocol tells each node in each round to listen or what message to broadcast based on the node's history. This history includes the messages the node received, when it received them and its initial private input. We assume that this private input includes the number of nodes, $n$, the maximum degree, $\Delta$, the receiver fault probability, $p$, a private random string, and any other data nodes have as input in a protocol. Formally, a \textbf{history} for node $v$ is an $H \in \mathcal{H}$ where $\mathcal{H}$ is all valid histories. $\mathcal{H} = \{(i, R) : i \in \mathcal{I}, R \subseteq M \times \mathbb{N}\}$ where $\mathcal{I}$ is all valid private inputs, $R$ gives what messages $v$ has received in each round and $M = \{0, 1\}^{O(\log n)}$ is all $O(\log n)$ bit messages a node could receive in a single round. Formally, a \textbf{protocol} $P$ of length $T$ is a function $P : V \times [T] \times \mathcal{H} \to \{ \textit{listen} \} \cup M$ where $P(v,t,H) = \textit{listen}$ indicates that $v$ listens in round $t$ and $P(v,t,H) = m \in M$ indicates that $v$ broadcasts message $m \in M$ in round $t$. We let $|P| := T$ stand for the length of a $T$ round protocol. In this paper we assume that $T$ is at most $\poly(n)$.


\textbf{Simulating a Protocol in the Noisy Setting } We say that protocol $P'$ successfully \textbf{simulates} $P$ if, after executing $P'$ in the noisy setting, every node in the network can reconstruct the messages it would receive if $P$ were run in the faultless setting. In particular, for any set of private input $I \in \mathcal{I}^{|V|}$---letting $H(v, I)$ be $v$'s history after running $P$ with private inputs $I$---it must hold that after running $P'$ with private inputs $I$, every $v$ can compute $H(v, I)$. Note that nodes running $P'$ can send messages not sent by $P$ or send messages in a different order than they do in $P$. We call $P$ the \textbf{original protocol} and $P'$ the \textbf{simulation protocol}. We measure the efficacy of $P'$ as the limiting ratio of $\frac{|P'|}{|P|}$ when $T$ is sufficiently large and $n$ goes to infinity, which we call the \textbf{multiplicative overhead} of a simulation.

\section{Techniques Overview and Our Formal Results}\label{sec:intution}
As earlier mentioned, we show how to simulate a faultless protocol $P$ in three noisy settings of increasing difficulty. Throughout our simulation results, we use the notion of a \textbf{virtual round} of node $v$, $t_v$, which tracks how many rounds of $P$ node $v$ has successfully simulated. We say that a node $v$ is \textbf{most delayed} in a set of nodes $U$ when $t_v \le \min_{w \in U} t_w$. All three simulations roughly work by having nodes first exchange their virtual rounds with their neighbors. Nodes then locally synchronize by simulating the virtual round of their most delayed neighbor. Our simulations differ in how each defines a virtual round and how each exchanges information regarding nodes' virtual rounds. 

As a warmup and to study what sort of overhead local synchronization enables, we consider the setting where nodes have access to ``local progress detection''. Recall that local progress detection gives nodes oracle access to the virtual rounds of their neighbors as well as when faults occur. We show the following theorem which demonstrates that a $O(\log \Delta)$ overhead is possible in this setting.%
\begin{restatable}{theorem}{cheatingSimulator}
	\label{thm:global-control-simulation}
	Let $P$ be a general protocol of length $T$ for the faultless radio network model. $P$ can be simulated in the noisy radio setting using local progress detection in $O(T \log \Delta + \log n + k)$ rounds with probability at least $1 - \exp(-k)$ for any $k \ge 0$ .
\end{restatable}%
\paragraph{Blaming Chain Intuition:} To prove the above theorem we use the idea of a blaming chain to argue that local synchronization enables small simulation overhead. Since every node simulates the round of its most delayed neighbor, a node $v$ will have all neighbors simulate its virtual round when its virtual round is minimal among virtual rounds of nodes in $\Gamma^{(2)}(v)$. Moreover, as we will show, once $v$ is most delayed in $\Gamma^{(2)}(v)$ it does not take \emph{too many} additional rounds for nodes to successfully simulate $v$'s virtual round. Thus, if $v$ takes many rounds to successfully simulate virtual round $x$, it must be because there was a $u \in \Gamma^{(2)}(v)$ that required many rounds to simulate virtual round $x-1$. In this way $v$ can blame its delay on $u$. Node $u$, in turn, can blame the fact that it required many rounds to simulate virtual round $x-1$ on the fact that one of its 2-hop neighbors, say $w$, required many rounds to simulate virtual round $x-2$ and so on. Thus, if node $v$ is very delayed we can explain this delay by some such blaming chain. There are exponentially many possible such blaming chains, but---by way of concentration inequalities we prove---we show that long blaming chains can be shown to have exponentially small probability and so a union bound over all long blaming chains will allow us to show that no long blaming chain occurs.

\medskip
We next show how to efficiently spread virtual round information without using progress detection. In particular, we show a $\poly(\log \Delta, \log \log n)$ overhead for non-adaptive protocols where nodes know a priori the rounds of the original protocol they are sent messages without collision.\footnote{When we write that an event occurs \textbf{with high probability (w.h.p.)}, we mean that it occurs with probability $1 - \frac{1}{n^c}$, and that the constant $c = \Theta(1)$ can be made arbitrarily large by changing the constants in the routines.}

\begin{restatable}{theorem}{staticSimulationTheorem}
	\label{thm:local-static-simulation}
	Let $P$ be a non-adaptive protocol of length $T$ for the faultless radio network model. $P$ can be simulated in $O((T + \log n) \log^3 \Delta \log \log n \log \log \log n)$ rounds in the noisy radio setting with high probability by \textsc{MainNonAdaptive}.
\end{restatable} 


\paragraph{Progress Throttling and Distributed Binary Search Intuition:} In addition to blaming chains, the main technique we use to prove the above result is progress throttling and a novel algorithm for binary search in noisy radio networks. Since in this setting nodes can no longer learn their neighbors' virtual rounds using local progress detection, our goal is to provide nodes an alternative means of learning their neighbors virtual rounds; namely, we use progress throttling and distributed binary search. Since virtual rounds can be as small as $1$ and as large as the length of the entire protocol which is as large as $\poly (n)$, the range over which we must binary search might seem to be as large as $\poly(n)$; a binary search over this range would require a prohibitive $O(\log n)$ iterations. For this reason, we slow down nodes that make progress too quickly by only increasing their virtual round at most once after $\log \Delta$ simulated rounds. We show that this throttling keeps all virtual rounds within an additive $O(\log n)$ range. This, in turn, allows our binary search to require only $O(\log \log n)$ iterations. 

Moreover, actually implementing a distributed binary search in a radio network presents technical challenges of its own. The natural solution we use for node $v$ to learn the smallest virtual round of a neighbor is for $v$ to repeatedly ask its neighbors if any of them have a virtual round below the midpoint of the binary search range. $v$ then updates its binary search parameters in the usual way. This strategy enables an efficient binary search in radio networks because to update its binary search parameters, $v$ need only hear from at most one neighbor. However, such a strategy suffers from the following problem of divergent ranges. Suppose node $v$ has a neighbor $u$. $u$  might have neighbors that are not neighbors of $v$ which influence the range over which $u$ searches. As such $u$ might end up searching over a different range than $v$. These divergent ranges may, unfortunately, render $v$'s binary search nonsensical: $v$ might query $u$ to learn if its virtual round is below $v$'s binary search midpoint and $u$ could respond with whether or not it is below the midpoint of an entirely different search range, namely $u$'s search range. The key idea we use to overcome this issue is to carefully mark nodes ``silent'' if they might interfere with another node's binary search. Specifically, we show that, given a node $v$ whose virtual round is minimal in $\Gamma^{(2)}(v)$,  a careful silencing of possibly deviating nodes causes nodes in $\Gamma(v)$ to always update their binary search parameters in exactly the same manner as $v$. Furthermore, we show that, while nodes in $\Gamma^{(2)}(v) \setminus \Gamma(v)$ might update their parameters differently, our silencing ensures that these nodes never interfere with the binary search performed by nodes in $\Gamma(v)$.

\medskip
Lastly, we describe how to simulate \emph{any} protocol with a $O(\Delta \log ^2 \Delta)$ multiplicative overhead.%
\begin{restatable}{theorem}{generalSimulationTheorem}
	\label{thm:local-general-simulation}
	Let $P$ be a general protocol of length $T$ for the faultless radio network model. \textsc{MainGeneral} simulates $P$ in the noisy radio setting in $O((T \log \Delta + \log n) \Delta \log \Delta)$ rounds w.h.p.
\end{restatable}
\paragraph{Token Exchange Intuition:} In addition to blaming chains, the main technique we use to show the above theorem is a token exchange strategy. Recall that the main challenge in the general setting is that even if a node knows its neighbors are simulating its virtual round, the node cannot tell if the absence of a message indicates that it receives no message in this round in the original protocol or that a random fault occurred. As such if a node does not have a message to deliver to its neighbor, we force it to send its neighbor a token indicating that it has no message to send. This allows $v$ to distinguish between a random fault and a round in which it simply receives no message.

\medskip
In \Cref{sec:lower-bounds} we give a formal statement of how simulations in two natural settings require $\Omega(\log \Delta)$ overhead but we give some intuition here as to why this might be true here.

\paragraph{Lower Bound Construction Intuition:} Consider a star with degree $\Delta$ where the center node wants to send $T$ messages to its neighbors. A noiseless protocol requires $T$ rounds but if the center only sends the original messages (and not e.g.\ an error correcting code) and random faults occur then the center must send each message about $\Omega(\log \Delta)$ times to guarantee all messages are delivered. Likewise consider a complete bipartite graph where each node on the left wants to deliver a message to the right. The existence of collisions in radio networks means that effectively only one node on the left can broadcast per round and, like the star, every node must broadcast about $\Omega(\log \Delta)$ times to deliver its message if there are random faults.

  



\section{Efficient Simulations for Noisy Radio Networks}
Having defined the noisy radio network model and given intuition, we now prove our simulations.

\subsection{Warmup: Simulation with Local Progress Detection}\label{sec:warmup-progress-detection}

We begin by giving our simulation with $O(\log \Delta)$ multiplicative overhead in the setting where nodes have access to \textbf{local progress detection}. 
Our main theorem of the section is as follows.

\cheatingSimulator*



To rigorously define local progress detection, we introduce our notion of how much simulation progress nodes have made, the \textbf{virtual round}. We say that a node $v$ successfully completed round $t$ if it has succeeded in taking the action in $P$ that it takes in $t$: $v$ successfully broadcasts its message $m$ of round $t$ if every node in $\Gamma(v)$ either has received $m$ or had a collision in the original protocol $P$ at round $t$; $v$ successfully completes its listening action of round $t$ if $v$ receives the message it receives in round $t$ of $P$ or a collision occurs at $v$ in round $t$ of $P$. We formally define the virtual round and local progress detection.

\begin{definition}[Virtual Round]
The virtual round of node $v \in V$ in a given round of simulation $P'$ is the smallest $t_v \in \mb{Z}_{\ge 1}$ such that $v$ has not successfully completed round $t_v$ of $P$.
\end{definition}
\begin{definition}[Local Progress Detection]
In the noisy radio network model with local progress detection every node $v$ knows the virtual round of every node $w \in \Gamma^{(2)}(v)$ in every simulation round.
\end{definition}
Notice that our definition of virtual round allows a node $v$ to learn when a fault occurs at $v$: if $\Gamma(v)$ simulates round $t_v$---which $v$ can compute since it knows the virtual round of nodes in $\Gamma^{(2)}(v)$---and $v$ listens but its virtual round does not increase then $v$ knows a receiver fault occurred at $v$.


We now prove the main theorem of this section.
\begin{proof}[Proof of \Cref{thm:global-control-simulation}]
	We begin by describing our simulation, $P'$. Each node $v$ repeatedly does the following in each round of $P'$. Let $u$ be the node in $\Gamma(v)$ with minimal $t_u$. 
	$v$ takes the action that it takes in round $t_u$ of $P$. That is, $v$ tries to ``help'' $u$ by simulating its virtual round. 
	
Let $v \in V$ and note that the virtual round $t_v$ never decreases. Hence for $x \in \mb{Z}_{\ge 1}$ we can define $D_{v, x}$ as the earliest round in $P'$ when $t_v \ge x$. Notice that $D_{v, T}$ just is the number of rounds that $v$ takes to simulate all rounds of the original protocol $P$. Thus, it will suffice for us to argue that $D_{v, T}$ is not too large for every $v$. We begin by finding a recurrence relation on $D_{v, x}$ saying that $v$ will advance soon after its 2-hop neighborhood $\Gamma^{(2)}(v)$ has virtual round at least $x$.
  
Consider when $t_v = x-1$. If the action of $v$ at round $x-1$ in $P$ is to broadcast, then once $v$'s virtual round is minimal in $\Gamma^{(2)}(v)$, all of its neighbors will simulate its round. This takes $\max_{w \in \Gamma^{(2)}(v)} D_{w, x-1}$ rounds. Once all nodes in $\Gamma^{(2)}(v)$ have virtual round at least $x-1$, every such node in $\Gamma(v)$ is simulating round $x-1$. By definition of $P'$, $v$ then broadcasts until all of its neighbors that are not incident to a collision in round $x-1$ of $P$ receive $v$'s message without a fault occurring. Any neighbor of $v$ that is sent a message without collision in round $x-1$ in $P$ will now take $\mc{G}(p)$ rounds, a geometric random variable with constant expectation. Therefore, the number of rounds until $v$ has sent all its relevant neighbors its message from round $x-1$ of $P$ is at most $Y_{v, x} := \max \{\mc{G}_1(p), \mc{G}_2(p), \ldots, \mc{G}_{|\Gamma(v)|}(p)\}$, where $\mc{G}_i$ are IID geometric random variables with constant expectation. Finally, we conclude that $D_{v, x} = \max_{w \in \Gamma^{(2)}(v)} \left( D_{w, x-1}\right) + Y_{v,x}$, again, where $Y_{v,x}$ is the maximum of at most $\Delta$ many geometric random variables with constant expectation.

  If the action of $v$ at round $x-1$ of $P$ is to receive a message without collision then, after every node in $\Gamma^{(2)}(v)$ has virtual round at least $x-1$, the neighbor of $v$ that sends $v$ a message without collision in round $x-1$ of $P$ will repeatedly send this message to $v$ until $v$ receives it without a receiver fault. It takes $\max_{w \in \Gamma^{(2)}(v)} D_{w, x-1}$ rounds of simulation for $v$ to have the smallest virtual round in its neighborhood and it takes $\mc{G}(p)$ rounds of simulation for $v$'s neighbor to send it a message. Thus, in this case we have $D_{v, x} = \max_{w \in \Gamma^{(2)}(v)} \left( D_{w, x-1}\right) + \mc{G}(p)$.

Since in the case when $v$ is broadcasting $D_{v, x}$ is larger than the case in which $v$ is listening, we have $D_{v, x} \le \left(\max_{w \in \Gamma^{(2)}(v)} D_{w, x-1}\right) + Y_{v,x}$ where $Y_{v,x}$ is the maximum of at most $\Delta$ many geometric random variables. 

We now prove a tail bound for $D_{v, x}$ by union bounding over all ``blaming chains.'' A blaming chain for $D_{v,x}$ is a sequence of nodes paired with simulated rounds that could explain why $D_{v,x}$ is as large as it is. In particular, a blaming chain $C$ for $D_{v,x}$ is a sequence of (node, round) tuples $(v_x, x), (v_{x-1}, x-1), \ldots, (v_1,1 )$ where $v_{i} \in \Gamma^{(2)}(v_{i-1})$ and $v_x = v$. We let $\mathcal{C}(D_{v,x})$ stand for all such blaming chains of $D_{v,x}$. We say $|C| := \sum_{(v_i, i) \in C} Y_{v_i, i}$ is the length of blaming chain $C$.  Note that $|\mathcal{C}(D_{v,x})|$ is at most $(\Delta^2)^{x-1}$.
 Also notice that $D_{v, x}$ just is the length of the longest blaming chain in $\mathcal{C}(D_{v, x})$. See \Cref{fig:BC}. Thus, we have that
\begin{figure}
	\centering
	\begin{subfigure}[t]{0.3\textwidth}\centering
		\includegraphics[scale=.23]{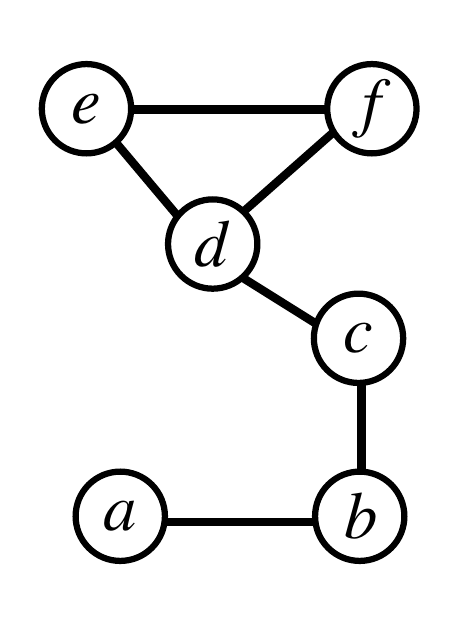}\caption{Network $G$}
	\end{subfigure}
	\begin{subfigure}[t]{0.49\textwidth}
		\includegraphics[scale=.3]{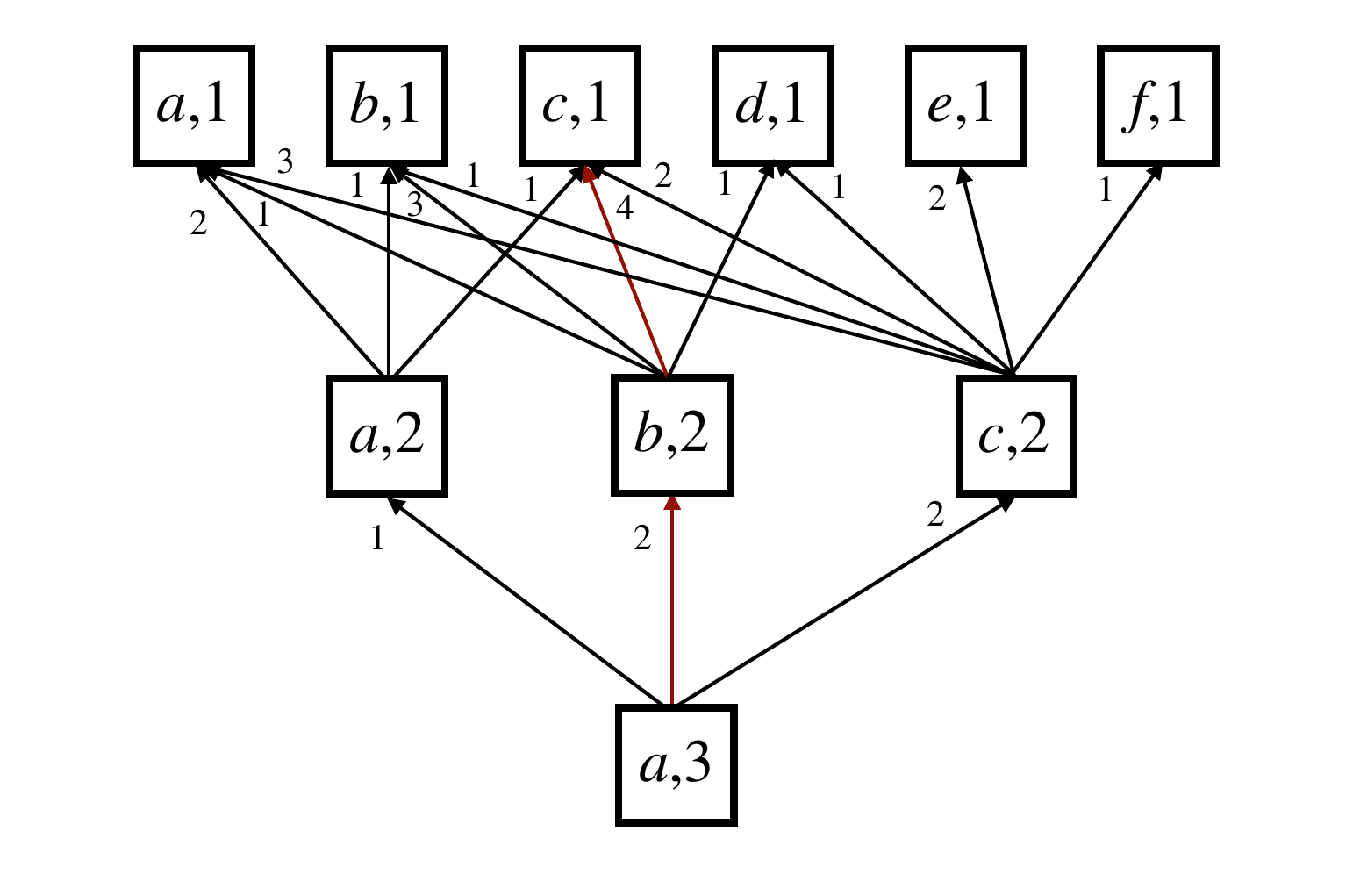}\caption{Blaming chains, $\mathcal{C}(D_{a, 2})$}
	\end{subfigure}
	\caption{$C \in \mathcal{C}(D_{a, 3})$: path away from $(a, 3)$. Edge pointing to $(v, x)$ labeled by the value of $Y_{v, x}$. $D_{a, 3} = \max_{C \in \mathcal{C}}(D_{a, 3}) |C| = 6$. Blaming chain that gives the value of $D_{a, 3}$: red.  }
	\label{fig:BC}
\end{figure}
\begin{align*}
D_{v, x} &= \max_{w \in \Gamma^{(2)}(v)} D_{w, x-1} + Y_{v, x}\\
& = \max_{C \in \mathcal{C}(D_{v,x})} \sum_{v_i \in C} Y_{v, i}\\
&= \max_{C \in \mathcal{C}(D_{v,x})}|C|
\end{align*}

\noindent As such, to show an upper bound on $D_{v, x}$ it suffices to show that no blaming chain is too long. Notice that the length of a blaming chain in $\mathcal{C}(D_{v,x})$ just is the sum of $x$ random variables each of which is the max of at most $\Delta$ geometric random variables with constant expectation. In \Cref{lem:chernoff-sum-of-max} of \Cref{sec:tail-bounds} we prove a Chernoff-style tail bound for such sums, showing that for any given $C \in \mathcal{C}(D_{v, x})$ we have that $\Pr[|C| \ge c(x \log \Delta + k)] \le \exp(-k)$ for some constant $c > 0$. Letting $x = T$ and taking a union bound over all $C \in \mathcal{C}(D_{v,T})$ tells us that for a fixed $v$ we have $\Pr[\exists C \in \mathcal{C}(D_{v,T}) \text{ s.t. } |C| \ge c(T \log \Delta + k)] \le |\mathcal{C}(D_{v,T})| \cdot \exp(-k) \leq (\Delta^2)^{T-1} \cdot \exp(-k)$. Thus with probability at most $(\Delta^2)^{T-1} \cdot \exp(-k)$ a chain in $\mathcal{C}(D_{v,x})$ is of length more than $c(T \log \Delta + k)$ and so $D_{v, T}$ is of length more than $c(T \log \Delta + k)$ with probability at most $(\Delta^2)^{T-1} \cdot \exp(-k)$.

Taking a union bound over every vertex we have that there exists a vertex $v$ such that $D_{v, T}$ is more than $c(T \log \Delta + k)$ with probability at most $n(\Delta^2)^{T-1} \cdot \exp(-k)$. Setting $k \gets \ln n + (2T-2) \ln \Delta + k'$ we have that there exists a vertex $v$ such that $D_{v, x}$ is more than $O(T \log \Delta + \log n)$ is at most $\exp(k')$. Thus, we have that every vertex with probability at least $1 - \exp(k')$ has completed the protocol after $O(T \log \Delta + \log n)$ rounds of $P'$. 
  
\end{proof}

\subsection{Simulation for Non-Adaptive Protocols}\label{sec:static-simulators}
We now give our simulation results for non-adaptive protocols with $\poly (\log \Delta, \log \log n)$ multiplicative round overhead. Formally a non-adaptive protocol is as follows.

\begin{definition}[Non-Adaptive Protocol]
 A non-adaptive protocol $P$ is a protocol in which every node can determine if it receives a message in each round of $P$ regardless of the private inputs of the nodes. Formally, each node $v$ is given a list, $M_v$, of the rounds in which it receives a message without collision in $P$.
\end{definition}
\noindent We let $M_v.\textsc{getNextRound}$ return the smallest round in $M_v$ such that for every $r \in M_v$ such that $r < M_v.\textsc{getNextRound}$, $v$ has received the message it receives in round $r$ of $P$. Since in a non-adaptive protocol nodes know in which rounds they receive messages in $P$, nodes in a simulation can locally compute $M_v.\textsc{getNextRound}$. Also notice that even though non-adaptive protocols assume nodes know a priori when they are supposed to receive messages, nodes are oblivious of the contents of these messages. We focus on non-adaptive protocols to focus on challenges 1 and 2 mentioned in \Cref{sec:intro}.

However, we also note that a number of well-known protocols are non-adaptive.
For instance, assuming nodes know the network topology and use public randomness, the optimal broadcast algorithm of \citet{gkasieniec2007faster} is non-adaptive. In particular, under these assumptions nodes can compute the broadcast schedule of their neighbors for this protocol. The same is true of the Decay protocol of~\citet{bar1992time}. Since broadcast is the most studied problem in radio networks, the fact that the state-of-the-art broadcast algorithm is non-adaptive would seem to make studying non-adaptive protocols worthwhile. Moreover, though requiring that nodes know the network topology and use public randomness may seem restrictive, in the context of simulations this assumption is actually quite weak: we can dispense with both assumptions by simply running any primitive that informs nodes of the network topology or shares randomness before our simulation is run at a one-time additive round overhead. Any primitive to learn the network topology or share randomness from the classic radio network setting coupled with the simple $O(\log n)$ simulation as described in \Cref{sec:intro} suffices here. This overhead is negligible if the simulated protocol is sufficiently long or we run many simulations on our network. Lastly, we note that it is often the case that nodes can even \emph{efficiently} compute the broadcast schedule of their neighbors---see \citet{haeupler2016faster}---and so this one-time cost is often quite small. Broadly speaking, then, any protocol in which knowing the topology and sharing randomness allows nodes to compute their neighbors' broadcast schedule is non-adaptive.


We now state the main theorem of this section and proceed to describe how we prove it.

\staticSimulationTheorem*

To prove our theorem we build upon the idea of the preceding section of using virtual rounds to locally synchronize nodes. However, in the current setting it is difficult for nodes to confirm when their broadcasts have succeeded, and so we must relax our definition of virtual rounds. In particular, for the remainder of this section we let the \textbf{virtual round} of each node be the minimum between the largest $t_v \in \mathbb{Z}_{\ge 1}$ such that $v$ receives a message without collision in $t_v$ and $v$ has successfully \emph{received} all messages it receives up to round $t_v - 1$ in $P$ in our simulation and a throttling variable $L$. That is, the virtual round of $v$ is $\min(M_v.\textsc{getNextRound}, L)$. An important property of virtual rounds in this setting is that given $L$ a node can always compute its virtual round. As discussed in \Cref{sec:intution}, this throttling will allow us to keep the range of our binary search small.

We now present the \textsc{MainNonAdaptive} simulation routine that simulates $P$ in the noisy setting (\Cref{A:mainStatic}). Its main subroutine, \textsc{LearnDelays} (\Cref{alg:LearnDelays}) performs a binary search to inform each node of the virtual round of its most delayed neighbor. In particular, in \textsc{LearnDelays} nodes with a virtual round less than the mean of their binary search range are ``active''. If there is an active node within 0 or 1 hop of node $v$ then node $v$ lowers the upperbound of its binary search. If there are no active nodes within $2$ hops of $v$ then it raises its binary search lowerbound. Otherwise there is an active node within 2 hops---which intuitively means that $v$ is not most delayed in $\Gamma^{(2)}(v)$---then node $v$ remains silent for the rest of the binary search so as to not interfere with the binary search of its neighbors. See \Cref{fig:BS} for an illustration of \textsc{LearnDelays}.

Lastly, nodes learn the distance of their nearest active nodes using helper subroutines (i) \textsc{DistToActive} (\Cref{A:DistToActive}) which, given a set of active nodes, checks for each $v$ if there is a node within distance $2$ of $v$ that is active, and (ii) \textsc{Broadcast} (\Cref{A:Broadcast}) which spreads a message to all neighbors of a node. 

We give each routine in pseudocode along with the lemma which gives its properties. 
Throughout our pseudocode we let $P(v, t)$ return the action taken by node $v$ in round $t$ of $P$. \fullOnly{We end this section with the proofs of our lemmas and theorem.}

\begin{algorithm}
  \caption{$\textsc{MainNonAdaptive}$ for node $v$}
  \label{A:mainStatic}
\begin{algorithmic}
  \State $t_v \gets 1$
  \For{$L = 1, 2, \ldots, T + O(\log n)$} \Comment{``Outermost iteration''}
    \For{repeat $O(\log \Delta)$ times} \Comment{``Innermost iteration''}
      \State $m_v \gets$ \textsc{LearnDelays}
      \State do action $P(v, m_v)$
      \If {$m_v = M_v.\textsc{getNextRound}$ and $v$ received a message}  
      \State $t_v \gets \min(L, M_v.\textsc{getNextRound})$ \Comment{Throttle $v$}
      \EndIf
     
    \EndFor
  \EndFor
\end{algorithmic}
\end{algorithm}

\begin{restatable}{lemma}{LmainRoutine}
  \label{L:main-routine}
  Assume that $L - O(\log n) \le
  t_v < L$ for all $v \in V$ and let $v$ be a most delayed node in its 2-hop neighborhood. After an innermost iteration of \textsc{MainNonAdaptive}, $t_v$ will increase by one with at least constant probability. Moreover, the running time of each innermost iteration is $O(\log^2 \Delta \log \log n \log \log \log n)$ rounds.
\end{restatable}

We now give our helper routine, \textsc{LearnDelays}, which informs nodes of their most delayed neighbors using a distributed binary search. 

\captionsetup[subfigure]{labelformat=empty}
\begin{figure}
	\centering
	\begin{subfigure}{0.24\textwidth}
		\includegraphics[scale=.23]{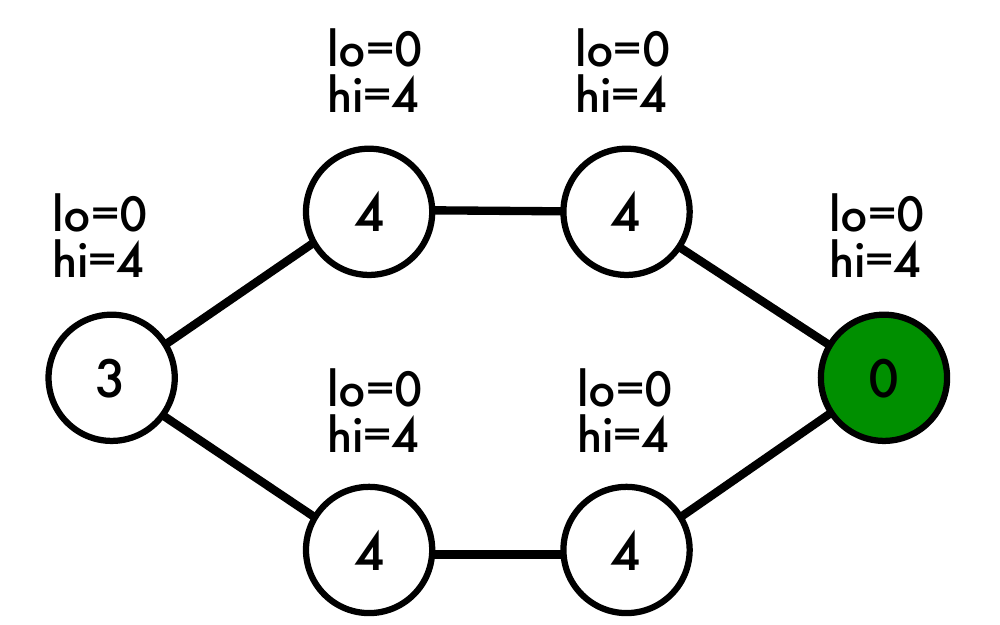}\caption{Iteration 0}
	\end{subfigure}
	\begin{subfigure}{0.24\textwidth}
		\includegraphics[scale=.23]{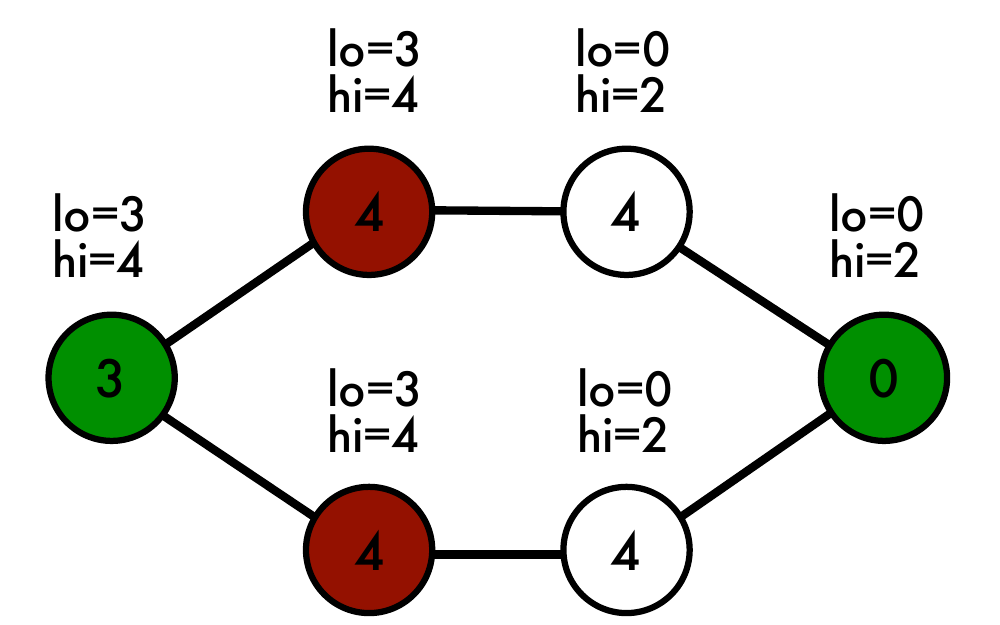}\caption{Iteration 1}
	\end{subfigure}
	\begin{subfigure}{0.24\textwidth}\centering
		\includegraphics[scale=.23]{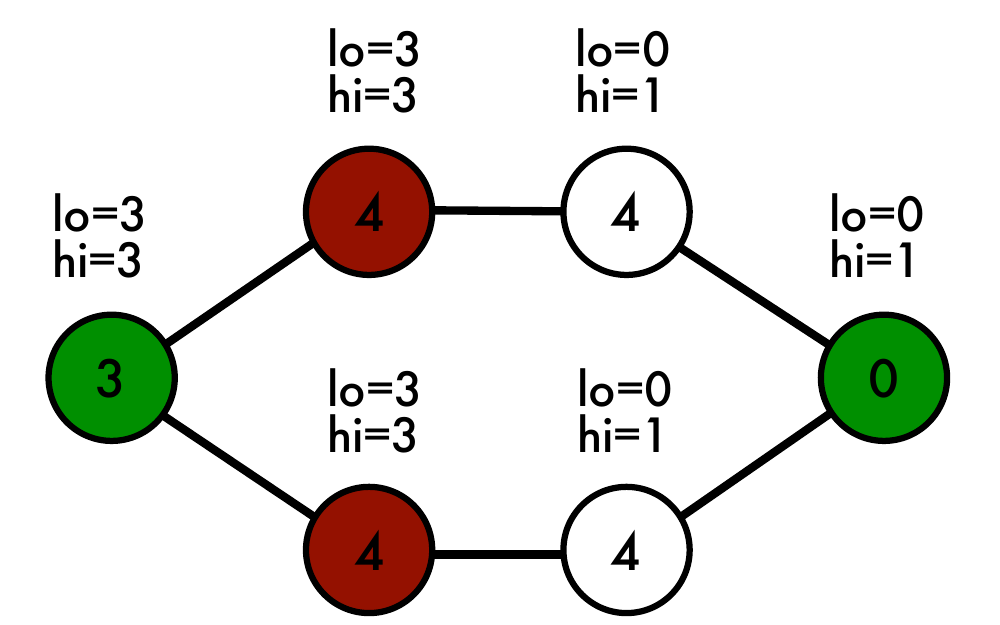}\caption{Iteration 2}
	\end{subfigure}
	\begin{subfigure}{0.24\textwidth}\centering
		\includegraphics[scale=.23]{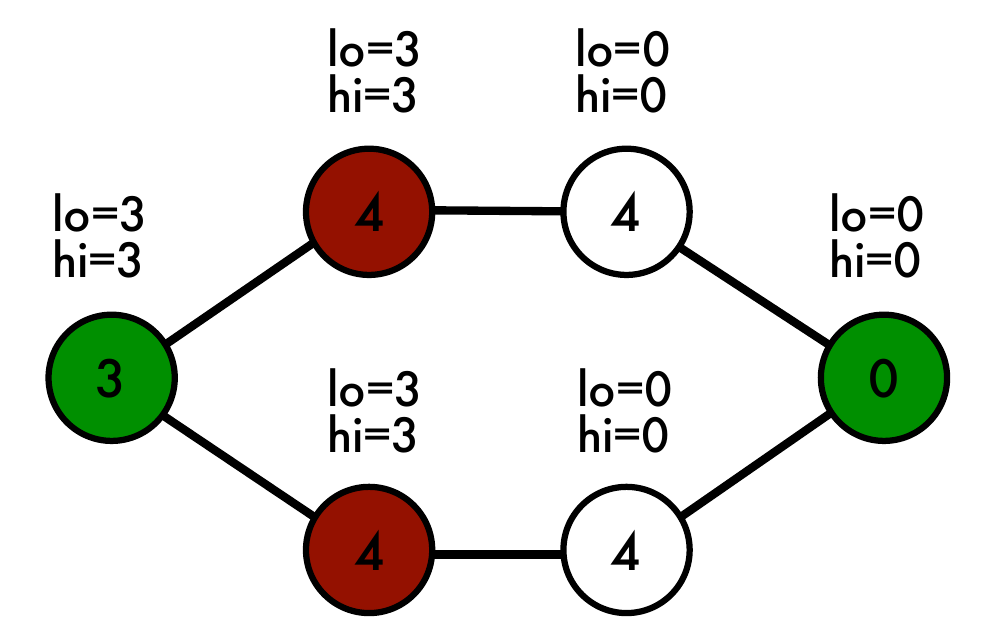}\caption{Iteration 3}
	\end{subfigure}
	\caption{\textsc{LearnDelays} assuming \textsc{DistToActive} always succeeds. Nodes labeled with rescaled virtual rounds (i.e.\ $v$ labeled with $t_v - L + O(\log n)$). Silent nodes: red. Active nodes: green. Far left and far right nodes: minimal in their two-hop neighborhoods so their neighbors learn their delays.}\label{fig:BS}
\end{figure}

\begin{algorithm}
\caption{\textsc{LearnDelays} for node $v$}
\label{alg:LearnDelays}
\begin{algorithmic}
  \Require{$L$, $t_v$}
  \State $lo \gets L - O(\log n)$; $hi \gets L$
  \While{$lo \neq hi$} \Comment{repeats $O(\log \log n)$ rounds} {
    \State $v$ is marked as "active" iff $t_v \le \lfloor \frac{lo+hi}{2} \rfloor$ and $v$ is not marked ``silent''
    \State dist $\gets$ \textsc{DistToActive}
    \State \textbf{if} dist is ``=2'' \textbf{then} mark $v$ as "silent" until the end of \textsc{LearnDelays}
    \If{dist is ``=0'' or ``=1''} \State $hi \gets \lfloor \frac{lo+hi}{2} \rfloor$
    \Else \State  $lo \gets \lfloor \frac{lo+hi}{2} \rfloor + 1$
    \EndIf
  }
  \EndWhile
  \State \Return{$lo$}
\end{algorithmic}
\end{algorithm}

\begin{restatable}{lemma}{LlearnDelays}
  \label{L:LearnDelays}
  Assume that $L - O(\log n) \le t_v \le L$ for all $v \in V$ and let $v$ be a most delayed node in its 2-hop neighborhood. After \textsc{LearnDelays} terminates, each 1-hop neighbor $w$ of $v$ will have $m_w = t_v$ with at least constant probability, i.e., every $w$ will try to help $v$ advance. The running time of \textsc{LearnDelays} is $O(\log^2 \Delta \log \log n \log \log \log n)$ rounds.
\end{restatable}

We now give \textsc{DistToActive}, the main helper routine for our distributed binary search routine.
\begin{algorithm}
  \caption{\textsc{DistToActive} for node $v$}
  \label{A:DistToActive}
\begin{algorithmic}
  \Require{is $v$ active?}
  \Ensure{``=0'', ``=1'', ``=2'', or ``>2'' as the smallest distance to an active node}
  \State \algorithmicif\ $v$ is active \algorithmicthen\ \textsc{Broadcast} an arbitrary message X \algorithmicelse\ stay silent
  \State \algorithmicif\ $v$ received X \algorithmicthen\ \textsc{Broadcast} an arbitrary message Y \algorithmicelse\ stay silent
  \State \algorithmicif\ $v$ is active \Return ``=0'' 
  \State \algorithmicif\ $v$ received X \Return ``=1'' 
  \State  \algorithmicif\ $v$ received Y \Return ``=2''
  \State \Return ``>2'' 
\end{algorithmic}
\end{algorithm}

\begin{restatable}{lemma}{LdistToActive}
  \label{L:DistToActive}
  Let $A \subseteq V$ be any set of active nodes. After \textsc{DistToActive} terminates, a fixed node $v$ correctly learns its distance to the nearest active node as 1,2, or more than 2 with probability at least $1 - O\left(\frac{1}{\Delta^2 (\log \log n)^2} \right)$. If $v$'s distance to the nearest active node is 0 then $v$ learns it as such with probability $1$. The running time of \textsc{DistToActive} is $O(\log^2 \Delta \log \log \log n)$ rounds.
\end{restatable}

Lastly, we give \textsc{Broadcast} which is based on the Decay protocol of~\citet{bar1992time} and spreads information among nodes.
\begin{algorithm}
  \caption{\textsc{Broadcast} for node $v$}
  \label{A:Broadcast}
\begin{algorithmic}
  \Require{a message to broadcast}
  \For{$O(\log \Delta \log \log \log n)$ times} \Comment{Outermost iteration}
  \For{$i = 1, 2, \ldots, O(\log\Delta)$} \Comment{Innermost iteration}
  \State $v$ broadcasts its message with probability $2^{-i}$
  \EndFor
  \EndFor
\end{algorithmic}
\end{algorithm}
\begin{restatable}{lemma}{Lbroadcast}
  \label{L:broadcast}
  Let $A \subseteq V$ be any set of nodes broadcasting the same message, while other nodes $V \setminus A$ are listening. After \textsc{Broadcast} terminates, a node $v$ with at least one broadcasting neighbor will receive the message with probability at least $1 - O\left(\frac{1}{\Delta^2 (\log \log n)^2} \right)$. Nodes without neighbors will not receive any messages. The running time of \textsc{Broadcast} is $O(\log^2 \Delta \log \log \log n)$ rounds.
\end{restatable}

\gdef\ProofLbroadcast{
\begin{proof}[Proof of \Cref{L:broadcast} (\textsc{Broadcast})] 
Call nodes in $A$ \textbf{informed} and let $m$ be the message known by nodes in $A$. Let the number of informed neighbors of $v$ be $[2^k, 2^{k+1}\rangle$ for $k \le \log_2 \Delta$. An easy calculation shows that during the innermost iteration when $i = k$, $v$ will receive the message with probability $\Omega(1) \cdot p$. Since $p = \Omega(1)$, we have proven $v$ receives $m$ with constant probability in each outermost iteration. There are $O(\log \Delta \log \log \log n)$ outermost iterations. Thus, the probability that $v$ never receives a message is at most $\exp(-O(\log \Delta \log \log \log n)) \le O\left(\frac{1}{ \Delta^2 (\log \log n)^2} \right)$. Lastly, a running time of $O(\log^2 \Delta \log \log \log n)$ follows by the definition of \textsc{Broadcast}.
\end{proof}
}\fullOnly{\ProofLbroadcast}

\gdef\ProofLdistToActive{
\begin{proof}[Proof of \Cref{L:DistToActive} (\textsc{DistToActive})]
Let $u$ be $v$'s nearest active node. The cases when $\dist(v, u) = 0$ and $\dist(v, u) > 2$ are trivial by definition of \textsc{DistToActive}. Consider when $\dist(v, u) = 2$. Let $w$ be a mutual neighbor of $v$ and $u$. The probability $v$ does not learn its distance as 2 is at most the probability that $w$ does not receive $X$ or $v$ does not receive $Y$. By \Cref{L:broadcast} and a union bound this probability is at most $O\left(\frac{1}{\Delta^2(\log \log n)^2} \right)$. Lastly, the case when $\dist(v, w) = 1$  holds for similar reasons. The running time of \textsc{DistToActive} follows from the running time of \textsc{Broadcast}.
\end{proof}
}\fullOnly{\ProofLdistToActive}

\gdef\ProofLlearnDelays{
\begin{proof}[Proof of \Cref{L:LearnDelays} (\textsc{LearnDelays})]

Let $v \in V$ be a most delayed node in its 2-hop neighborhood. The main idea of this proof is to show that nodes in $\Gamma(v)$ always update their binary search parameters in exactly the same manner as $v$. Furthermore, we show that, while nodes in $\Gamma^{(2)}(v) \setminus \Gamma(v)$ might update their parameters differently, these nodes never interfere with the binary search performed by nodes in $\Gamma(v)$.

By $\Cref{L:DistToActive}$ and a union bound over $v$'s at most $\Delta^2$ 2-hop neighbors and the $O(\log \log n)$ rounds of \textsc{LearnDelays}, we have that every 2-hop neighbor of $v$ correctly learns its distance to an active node with constant probability in every round of \textsc{LearnDelays}. From this point on in the proof we condition on every 2-hop neighbor of $v$ correctly learning its distance to an active node in every iteration of \textsc{LearnDelays}. We say that nodes $v$ and $w$ \textbf{deviate} if at some point $v$ and $w$ have different values for $lo$ or $hi$. For any variable $\alpha$ of \textsc{LearnDelays} we let $\alpha_i(v)$ stand for $v$'s value of $\alpha$ in iteration $i$. For example, $hi_i(v)$ is $v$'s value of $hi$ in iteration $i$ of \textsc{LearnDelays}.
  
We show that three properties hold in every iteration $i$ of \textsc{LearnDelays} by induction over $i$: (1) $v$ is not marked silent; (2) $u \in \Gamma(v)$ does not deviate from $v$; (3) if $w \in \Gamma^{(2)}(v) \setminus \Gamma(v)$ deviates from $v$ in iteration $i$ then $w$ is not active in iteration $j$ for any $j > i$. (1) and (2) will show that nodes in $\Gamma(v)$ update their binary search parameters appropriately and (3) will show that any node which fails to do so will never interfere with the binary search of nodes in $\Gamma(v)$.

Our base case for the $0$th iteration is trivial. (1) holds since $v$ is not initially marked silent. (2) holds since every node initializes $hi$ and $lo$ to the same values and so $u \in \Gamma(v)$ initializes $hi$ and $lo$ to the same values as $v$. (3) holds since no $w$ deviates from $v$ initially.

We now show our inductive step. Assume as an inductive hypothesis that in every iteration $i' < i$ we have that (1), (2) and (3) hold. We now show that (1), (2) and (3) hold in iteration $i$. 
\begin{enumerate}[(1)]
\item We would like to show that $v$ is not marked silent in iteration $i$. Assume for the sake of contradiction that $v$ is marked silent. $v$ being marked silent means that $v$ is not active, meaning $t_v > \lfloor  \frac{lo_i(v) + hi_i(v)}{2} \rfloor$, but that there is a node in $w \in \Gamma^{(2)}(v) \setminus \Gamma(v)$ that is active, meaning $t_w \leq \lfloor \frac{lo_i(w) + hi_i(w)}{2} \rfloor$. By (3) of our inductive hypothesis we know that $w$ has not deviated from $v$, meaning $\frac{lo_i(w) + hi_i(w)}{2}= \frac{lo_i(v) + hi_i(v)}{2}$ and so we then have 

\begin{align*}
t_w  &\leq \lfloor  \frac{lo_i(w) + hi_i(w)}{2} \rfloor\\
&= \lfloor \frac{lo_i(v) + hi_i(v)}{2} \rfloor\\
& < t_v
\end{align*}
Thus, $t_w < t_v$. However, since $t_w \in \Gamma^{(2)}(v)$ this contradicts the fact that $t_v$ is most delayed in its 2-hop neighborhood. Thus, we conclude that, in fact, in iteration $i$ it holds that $v$ is not marked silent.
\item We would like to show that in iteration $i$ it holds that $u \in \Gamma(v)$ does not deviate from $v$, i.e.\ $hi_{i+1}(u)= hi_{i+1}(v)$ and $lo_{i+1}(u) = lo_{i+1}(v)$. We case on whether or not in iteration $i$ we have $t_v \leq \lfloor \frac{lo_i(v) + hi_i(lo_i(v))}{2} \rfloor$. 

\begin{itemize}
\item Suppose that $t_v \leq \lfloor \frac{lo_i(v) + hi_i(v)}{2} \rfloor$. By (1) of our induction hypothesis $v$ has not been marked silent and so in this case $v$ is active. If $v$ is active then for any $u \in \Gamma(v)$ (where potentially $u = v$ since $v \in \Gamma(v)$) we have $\text{dist}_i(u) \leq 1$ and so $hi_{i+1}(u) \gets \lfloor \frac{lo_i(u) + hi_i(u)}{2} \rfloor$ and $lo_{i+1}(u) \gets lo_i (u)$. By (2) of our induction hypothesis we have that $lo_i(u) = lo_i(v)$ and $hi_i(v) = hi_i(u)$ for every $u \in \Gamma(v)$. Thus, for $u \in \Gamma(v)$ we have $lo_{i+1}(u) \gets lo_i (u) = lo_i(v) = lo_{i+1}(v)$ and $hi_{i+1}(u) \gets \lfloor \frac{lo_i(u) + hi_i(u)}{2} \rfloor = \lfloor \frac{lo_i(v) + hi_i(v)}{2} \rfloor = hi_{i+1}(v)$. Thus, in this case $v$ and $u$ do not deviate.

\item Suppose that $t_v > \lfloor \frac{lo_i(v) + hi_i(v)}{2} \rfloor$ in iteration $i$. Combining the fact that $v$ is most delayed in its two-hop neighborhood, with (2) of our inductive hypothesis it follows that $\text{dist}_i(v) \geq 2$.

Now assume for the sake of contradiction that $u$ deviates from $v$ in this iteration. That is, $\text{dist}_i(u) \leq 1$. Since $\text{dist}_i(v) \geq 2$ we know that $u$ must deviate from $v$ because there is an active node in $w \in \Gamma^{(2)}(v) \setminus \Gamma(v)$. However, by (3) of our inductive hypothesis any such node must have not deviated from $v$. Since $v$ is most delayed in its two-hop neighborhood it follows that $t_w \geq t_v \lfloor \frac{lo_i(v) + hi_i(v)}{2} \rfloor = \lfloor \frac{lo_i(w) + hi_i(w)}{2} \rfloor$ which contradicts the fact that $w$ is active. Thus, in this case $u$ must not deviate from $v$.
\end{itemize}

\item We would like to show that if $w \in \Gamma^{(2)}(v) \setminus \Gamma(v)$ deviates from $v$ in iteration $i$ then $w$ is not active in iteration $j$ for any $j > i$. We demonstrate this by showing that any such $w$ is either marked silent in iteration $i$ or in every subsequent iteration $t_w$ is outside of the range over which $w$ is searching. Consider a $w$ that deviates from $v$. $w$ deviates if $\text{dist}_i(v) \leq 1$ but $\text{dist}_i(w) \geq 2$ or $\text{dist}_i(w) \leq 1$ but $\text{dist}_i(v) \geq 2$. Consider each case.
\begin{itemize}
\item Suppose $\text{dist}_i(v) \leq 1$ but $\text{dist}_i(w) \geq 2$. First notice that in this case $v$ is active since $v$ is minimal in its two-hop neighborhood and by (1) of our inductive hypothesis $v$ is not silent. Thus, since $v$ is active and $v$ is at most two hops from $w$ we have $\text{dist}_i(w) \leq 2$, meaning that $\text{dist}_i(w) = 2$. However, this means that $w$ is marked as silent and will therefore never be active in any iteration $j > i$.
\item Suppose $\text{dist}_i(w) \leq 1$ but $\text{dist}_i(v) \geq 2$. By (3) of our inductive hypothesis we have that $lo_i(v) = lo_i(w)$ and $hi_i(v) = hi_i(w)$. Moreover, since $\text{dist}_i(v) \geq 2$ we know that in iteration $i$ it holds that $t_v > \lfloor \frac{lo_i(v) + hi_i(v)}{2} \rfloor$. Since $v$ is minimal in its two-hop neighborhood we have that $t_v \leq t_w$. Combining these facts we have that 

\begin{align*}
t_w &\geq t_v\\
&> \left\lfloor \frac{lo_i(v) + hi_i(v)}{2} \right\rfloor\\
&= \left\lfloor \frac{lo_i(w) + hi_i(w)}{2} \right\rfloor\\
\end{align*}

Thus, $t_w > \left\lfloor \frac{lo_i(w) + hi_i(w)}{2} \right\rfloor$. Moreover, since $\text{dist}_i(w) \leq 2$ we know that $hi_{i+1}(w) = \left\lfloor \frac{lo_i(w) + hi_i(w)}{2} \right\rfloor$ and so $t_w > hi_{i+1}(w)$. Since $hi$ is non-increasing over iterations---i.e.\ $hi_{j+1}(w) \leq hi_{j}(w)$ for $j$ and any $w$---we will have that in every iteration $j > i$ it holds that $t_w > hi_{i+1}(w) \geq hi_j(w) \geq \lfloor \frac{lo_j(w) + hi_j(w)}{2} \rfloor$ meaning that $w$ cannot be active in iteration $j$ since $w$ is only active in iteration $j$ when $t_w \leq \lfloor \frac{lo_j(w) + hi_j(w)}{2} \rfloor$.

\end{itemize}

\end{enumerate}

Having shown our induction we now argue that the value that $v$ will learn from its binary search is $t_v$, it's virtual round. In particular, we argue that $v$ always updates its binary search parameters according to $t_v$. That is, in iteration $i$ we have $v$ reduces $hi$ when $t_v \leq \lfloor \frac{lo_i(v) + hi_i(v)}{2} \rfloor$ and $v$ increases $lo$ when $t_v > \lfloor \frac{lo_i(v) + hi_i(v)}{2} \rfloor$.
\begin{itemize}
\item If $v$ updates $hi$ in iteration $i$ then there must be some node $u \in \Gamma(v)$ which is active. That is, $t_u \leq \lfloor \frac{lo_i(u) + hi_i(u)}{2} \rfloor$. By (2) we have $lo_i(u) = lo_i(v)$ and $hi_i(u) = hi_i(v)$ and by the fact that $v$ is most delayed in its two-hop neighborhood we know that $t_v \leq t_u$. Thus $t_v \leq t_u \leq  \lfloor \frac{lo_i(u) + hi_i(u)}{2} \rfloor = \lfloor \frac{lo_i(v) + hi_i(v)}{2} \rfloor$.
\item If $v$ updates $lo$ then no node in $\Gamma(v)$ is active and so certainly $v$ itself is not active. But since $v$ is never marked silent by (1) it must not be active because $t_v > \lfloor \frac{lo_i(v) + hi_i(v)}{2} \rfloor$.
\end{itemize}
In either case we have that $v$ always updates its binary search parameters according to $t_v$ and so we have that $v$ will return $t_v$ from its binary search.

Thus, $v$ learns as its binary search value $t_v$ but since no node in $\Gamma(v)$ deviates from $v$ so too does every node in $\Gamma(v)$. Recalling that we conditioned on the fact that every node in $\Gamma^{(2)}(v)$ always learns its distance to an active neighbor correctly which happens with constant probability, we conclude our claim. Lastly, our running time comes from straightforwardly summing the running time of subroutines and noting that by assumption every value we are binary searching over lies within a $O(\log n)$ range.
\end{proof}
} \fullOnly{\ProofLlearnDelays}

\gdef\ProofLmainRoutine{
\begin{proof}[Proof of \Cref{L:main-routine} (\textsc{MainNonAdaptive})]
Fix a $v$ that is most delayed in its 2-hop neighborhood such that $t_v < L$. Note that by definition of a virtual round and the fact that $t_v < L$ we have that $P(v, t_v)$ is a listening action where in round $t_v$ of $P$ it holds that $v$ is sent a message without collision. By \Cref{L:LearnDelays} after \textsc{LearnDelays} terminates with constant probability each 1-hop neighbor $w$ will have $m_w = t_v$. Thus, every node in the 1-hop neighborhood of $v$ will simulate round $t_v$, i.e. the one neighbor that has a message for $v$ will broadcast its message and the rest of $v$'s neighbors will be silent. This message will be successfully received with probability $p$ which is constant by assumption and so $t_v$ will be incremented with constant probability. The running time follows by summing subroutines.
\end{proof}
}\fullOnly{\ProofLmainRoutine}

\gdef\ProofStaticSimulationTheorem{
\begin{proof}[Proof of \Cref{thm:local-static-simulation}]
  Let $Q := O(\log n)$. We will prove the following subclaim: Assume that for every vertex $v$ we have $L - Q \le t_v \le L$; letting $t_v'$ be $v$'s virtual round after $Q$ outermost iterations of \textsc{MainNonAdaptive} we have $L \le t'_v \le L + Q$ w.h.p. 
 
Note that this subclaim immediately proves the theorem. The length of $P$ is $\poly(n)$ by assumption and so we can divide $P$ into $\poly(n)$ chunks of rounds each of size $Q$. By induction and a union bound over chunks  we have that after $KQ$ outermost iterations of \textsc{MainNonAdaptive} it holds that for every node $v$ we have $QK \leq t_v \leq QK + 1$ w.h.p. Letting $K = \frac{|P|}{Q}$ yields that with high probability after $T + Q$ rounds every node $v$ is such that $|P| \leq t_v$, yielding our theorem.

We now prove the subclaim. Note that the assumption of our subclaim is that the preconditions of \Cref{L:main-routine} are satisfied. For simplicity and without loss of generality we assume $L = 0$. Let $x \in \mb{Z}_{\ge 1}$ and define $D_{v, x}$ as the random variable giving the index of the earliest innermost iteration of \textsc{MainNonAdaptive} when $t_v \ge x$. Using \Cref{L:main-routine} we get the following following recurrence relation. $D_{v, x} \leq \min_{w \in \Gamma^{(2)}(v)} D_{w, x-1} + \mc G_{v, x}(q)$, where $q$ is the constant probability of an inner iteration of \textsc{MainNonAdaptive} advancing a node most delayed in its 2-hop neighborhood given by \Cref{L:main-routine}.

We now prove a tail bound for $D_{v, x}$ by union bounding over all ``blaming chains.'' A blaming chain for $D_{v,x}$ is a sequence of nodes paired with simulated rounds that could explain why $D_{v,x}$ took as long as it did. In particular, a blaming chain $C$ for $D_{v,x}$ is a sequence of (node, round) tuples $(v_x, x), (v_{x-1}, x-1), \ldots, (v_1,1 )$ where $v_{i} \in \Gamma^{(2)}(v_{i-1})$ and $v_x = v$. We let $\mathcal{C}(D_{v,x})$ stand for such blaming chains. We say $|C| := \sum_{(v_i, i) \in C} \mathcal{G}_{v_i, i}$ is the length of blaming chain $C$.  Note that the size of $\mathcal{C}(D_{v,x})$ is at most $(\Delta^2)^{x-1}$.
 Also notice that $D_{v, x}$ just is the length of the longest blaming chain in $\mathcal{C}(D_{v, x})$. Thus, we have that
\begin{align*}
D_{v, x} &= \max_{w \in \Gamma^{(2)}(v)} D_{w, x-1} + \mathcal{G}_{v, x}\\
& = \max_{C \in \mathcal{C}(D_{v,x})} |C|\\
\end{align*}
\noindent As such, to show an upper bound on $D_{v, x}$ it suffices to show that no blaming chain is too long. Notice that the length of a blaming chain in $\mathcal{C}(D_{v,x})$ just is the sum of $x$ geometric random variables with constant expectation. \Cref{lem:chernoff-sum-of-max} of \Cref{sec:tail-bounds} gives a Chernoff-style tail bound for such sums and asserts that for any fixed blaming chain $C$ in $\mathcal{C}(D_{v, Q})$ we have $\Pr[|C| \ge c(Q \log \Delta + k)] \le \exp(-k)$ for some constant $c > 0$. A union bound over all blaming chains of all nodes gives us that $\Pr[\exists v \in V,\, D_{v, Q} \ge c(Q \log \Delta + k)] \le n (\Delta^2)^{Q-1} \exp(-k)$. Setting $k \gets 2Q \ln \Delta + O(\ln n)$, all $D_{v, Q} = O(Q \log \Delta + \log n) = O(\log n \log \Delta)$ w.h.p.  Note that $D_{v, Q}$ counts innermost iterations, hence it takes $\frac{D_{v, Q}}{c \log \Delta}$ outermost iterations for the condition to be satisfied. Moreover, $c$ is the constant in the iteration range of the innermost for loop of \textsc{MainNonAdaptive} and $Q$ does not depend on it. Hence we can set $c > 0$ sufficiently large such that $\frac{D_{v, Q}}{c \log \Delta} = \frac{O(\log n)}{c} \le Q$. Having shown that the number of outer iterations for every variable to arrive at virtual round $Q$ is at most $Q$, we conclude the subclaim. As we earlier argued, this gives our theorem. Note that our running time follows from summing the runtimes of our subroutines.
  
\end{proof}
}\fullOnly{\ProofStaticSimulationTheorem} 

\subsection{General Protocol Simulation} \label{sec:general-simulator}
Here, we provide our results for arbitrary protocols.  Our approach gives a $O(\Delta \log ^2 \Delta)$ multiplicative overhead. 

\generalSimulationTheorem*

Again, we build on the notion of a virtual round for this setting. Our main challenge in this setting is that even if a node knows its neighbors are simulating its virtual round, the node cannot tell if the absence of a message indicates that it receives no message in this round in the original protocol or that a random fault occurred. As such, in order for node $v$ to advance its virtual round after hearing no messages from a neighbor, $v$ must confirm that every neighbor was silent in the simulated round. Let $P$ be the original protocol for the faultless setting. Define the \textbf{token} for a node $v$ in round $r$ of $P$ to be either the message that $v$ is sending in round $r$ of $P$ or an arbitrary message indicating ``$v$ is not broadcasting'' if $v$ is silent in round $r$ of $P$. Next, we (re-)define the \textbf{virtual round} of a node $v$ to be the largest $t_v \in \mathbb{Z}_{\ge 1}$ such that $v$ successfully received \emph{all tokens from all neighbors} for rounds $1, 2, \ldots, t_v - 1$ (for a total of up to $(t_v-1)\Delta$ tokens).

Our simulation algorithm works in two phases: every node first informs its neighbors of its virtual round; next, nodes help the neighbor with the smallest $t_v$ they saw by sharing the token for that round. We now present pseudocode for the \textsc{ShareKnowledge} routine which shares messages from a node with all of its neighbors and \textsc{MainGeneral} which simulates $P$ in the noisy setting.

\begin{algorithm}[H]
\caption{\textsc{ShareKnowledge} for node $v$}
\begin{algorithmic}
  \Require{a message, $B_v$, that $v$ wants to share}
  \For{$O(\Delta \log \Delta)$ rounds}
  \State $v$ broadcasts $B_v$ with probability $\frac{1}{\Delta}$, independently from other nodes
  \EndFor
\end{algorithmic}
\end{algorithm}

\begin{lemma}\label{lem:shareKnowledge}
  After \textsc{ShareKnowledge} terminates, a fixed node $v$ successfully receives messages from all its neighbors with probability at least $3/4$ and successfully sends its message to all neighbors with probability at least $3/4$. The running time of \textsc{ShareKnowledge} is $O(\Delta \log \Delta)$ rounds.
\end{lemma}
\begin{proof}
Fix an arbitrary node $v$. Consider the event where $v$ receives a message from a fixed neighbor $w$ in a fixed iteration of \textsc{ShareKnowledge}. This event occurs iff $w$ broadcasts and all other neighbors of $v$, namely $\Gamma(v) \setminus \{w, v\}$, do not broadcast. This occurs with probability that is at least $\frac{1}{\Delta}(1 - \frac{1}{\Delta})^{|\Gamma(v) \setminus \{w, v\}|} \geq \frac{1}{\Delta}(1 - \frac{1}{\Delta})^{\Delta} \geq \Omega(\frac{1}{\Delta})$ by $(1-\frac{1}{x})^x = \Omega(1)$. The probability that $v$ does not hear from $w$ after $O(\Delta \log \Delta)$ iterations is $(1 - \Omega(\frac{1}{\Delta}))^{\Delta \log \Delta} \le \exp(-\Omega(\log \Delta)) \le \frac{1}{4 \Delta}$. Union bounding over all $|\Gamma(v)| \le \Delta$ possibilities for $w$ we get that the probability of $v$ not sharing knowledge with all neighbors is at most $1/4$.
\end{proof}

\begin{algorithm}[H]
\caption{\textsc{MainGeneral} for node $v$}
\begin{algorithmic}
  \State $t_v \gets 1$
  \For{$O(T \log \Delta + \log n)$}
  \State $v$ runs \textsc{ShareKnowledge}($t_v$)
  \State $m_v \gets$ smallest value $v$ receives from all nodes running \textsc{ShareKnowledge}
  \State \textsc{ShareKnowledge}(token for virtual round $m_v$)
  \State update $t_v$ if $v$ received all tokens for round $t_v$
  \EndFor
\end{algorithmic}
\end{algorithm}

\begin{lemma}
  \label{L:MainGeneral}
  Let $v$ be a most delayed node in its 2-hop neighborhood. After one \textsc{MainGeneral} loop iteration, $t_v$ will increase by one with at least constant probability.
\end{lemma}
\begin{proof}
Fix a node $v$ that is most delayed in its 2-hop neighborhood. By \Cref{lem:shareKnowledge} and a union bound over $v$'s at most $\Delta$ neighbors, the probability that the first call to \textsc{ShareKnowledge} successfully shares $v$'s message to all of its neighbors, and that all of $v$'s neighbors successfully share their tokens with $v$ is at least $\frac{1}{2}$. Since $v$ is minimal in its 2-hop neighborhood it follows that $m_w = t_v$ for $w \in \Gamma(v)$. Thus, $v$ will receive a token from every neighbor for its virtual round with probability at least $\frac{1}{2}$, meaning that $v$ increments its virtual round by one with constant probability.
\end{proof}

\begin{proof}[Proof of \Cref{thm:local-general-simulation}]
  We use the same blaming chain proof technique as in \Cref{thm:local-static-simulation} but use \Cref{L:MainGeneral} instead of \Cref{L:main-routine}. The proof is completely analogous as in \Cref{thm:local-static-simulation} but we include the proof here for completeness.
  
  Let $x \in \mb{Z}_{\ge 1}$ and define $D_{v, x}$ as the random variable giving the index of the earliest iteration of \textsc{MainGeneral} when $t_v \ge x$. Using \Cref{L:MainGeneral} we get the following following recurrence relation. $D_{v, x} \leq \min_{w \in \Gamma^{(2)}(v)} D_{w, x-1} + \mc G_{v, x}(q)$, where $q$ is the constant probability of an inner iteration of \textsc{MainNonAdaptive} advancing a node most delayed in its 2-hop neighborhood given by \Cref{L:MainGeneral}.

  We now prove a tail bound for $D_{v, x}$ by union bounding over all ``blaming chains.'' A blaming chain for $D_{v,x}$ is a sequence of nodes paired with simulated rounds that could explain why $D_{v,x}$ took as long as it did. In particular, a blaming chain $C$ for $D_{v,x}$ is a sequence of (node, round) tuples $(v_x, x), (v_{x-1}, x-1), \ldots, (v_1,1 )$ where $v_{i} \in \Gamma^{(2)}(v_{i-1})$ and $v_x = v$. We let $\mathcal{C}(D_{v,x})$ stand for such blaming chains. We say $|C| := \sum_{(v_i, i) \in C} \mathcal{G}_{v_i, i}$ is the length of blaming chain $C$.  Note that the size of $\mathcal{C}(D_{v,x})$ is at most $(\Delta^2)^{x-1}$.
  Also notice that $D_{v, x}$ just is the length of the longest blaming chain in $\mathcal{C}(D_{v, x})$. Thus, we have that
  \begin{align*}
  D_{v, x} &= \max_{w \in \Gamma^{(2)}(v)} D_{w, x-1} + \mathcal{G}_{v, x}\\
  & = \max_{C \in \mathcal{C}(D_{v,x})} |C|\\
  \end{align*}
  \noindent As such, to show an upper bound on $D_{v, x}$ it suffices to show that no blaming chain is too long. Notice that the length of a blaming chain in $\mathcal{C}(D_{v,x})$ just is the sum of $x$ geometric random variables with constant expectation. \Cref{lem:chernoff-sum-of-max} of \Cref{sec:tail-bounds} gives a Chernoff-style tail bound for such sums and asserts that for any fixed blaming chain $C$ in $\mathcal{C}(D_{v, T})$ we have $\Pr[|C| \ge c(T \log \Delta + k)] \le \exp(-k)$ for some constant $c > 0$. A union bound over all blaming chains of all nodes gives us that $\Pr[\exists v \in V,\, D_{v, T} \ge c(T \log \Delta + k)] \le n (\Delta^2)^{T-1} \exp(-k)$. Setting $k \gets 2T \ln \Delta + O(\ln n)$, all $D_{v, T} = O(T \log \Delta + \log n)$ w.h.p. Thus every node will achieve virtual round $T$ within $O(T \log \Delta + \log n)$ iterations of \textsc{MainGeneral} w.h.p.\ and so therefore within $O((T \log \Delta + \log n) \Delta \log \Delta)$ rounds.
\end{proof}

\section{Lower bounds}\label{sec:lower-bounds}

In this section we argue that an $\Omega(\poly \log \Delta)$ multiplicative overhead in simulation is necessary in two natural settings. In the first setting the simulation is not permitted to use network coding \citep{Ahlswede2006}. In the second setting the simulation must respect information flow in the sense that we show a lower bound when the graph is directed. It follows that any simulation with constant overhead either uses network coding or does not respect information flow. We also give a construction which we believe could be used to show an $\Omega(\poly \log \Delta)$ multiplicative overhead in simulation, even without any assumptions.

Additionally, we strengthen our lower bounds by proving them in the setting in which the simulation is granted ``global control''. Informally, global control eliminates the need for control messages by providing a centralized scheduler that can synchronize nodes based on how faults occur. The scheduler, however, cannot read the actual contents of the messages.

\begin{definition}[Global Control]
We say that a noisy radio network has global control when (1) nodes know the network topology, (2) nodes learn which nodes broadcast in each round and at which nodes receiver faults occur in each round, and (3) all nodes have access to public randomness.
\end{definition}

\noindent It is not difficult to see that access to global control is sufficient to achieve the local progress detection of \Cref{sec:warmup-progress-detection}. Moreover, notice that our simulation from \Cref{sec:warmup-progress-detection} with $O(\log \Delta)$ multiplicative overhead is both non-coding and respects information flow. As such, it is optimal for both settings.


\subsection{Non-Coding Simulations}
We now define a non-coding protocol and prove that simulations that do not use network coding suffer a $\Omega(\log \Delta)$ multiplicative overhead.

\begin{definition}[Non-Coding]
  We say that a simulation $P'$ in the noisy setting, which is simulating a protocol $P$ in the faultless setting,  is non-coding if any message sent in $P'$ is also sent in $P$ (though possibly by different nodes).
\end{definition}

We consider an isolated star with degree $\Delta$ where the center node wants to send $T$ messages to its neighbors. One can achieve a constant multiplicative overhead on this protocol by using an error correction code like Reed-Solomon~\cite{WickerRSCodes}. However, the following lemma shows that if coding is not used no such overhead is possible.
\begin{lemma}
  \label{L:non-coding}
  For any $T \ge 1$ and sufficiently large $\Delta$ there exists a faultless protocol of length $T$ on the star network with $\Delta + 1$ nodes such that any non-coding simulation of the protocol in the noisy setting with constant success probability requires $\Omega(T \log \Delta)$ many rounds even if the simulation has access to global control.
\end{lemma}
\begin{proof}
  As noted, the network is a star with $\Delta$ leafs. Let $r$ be the central node of the star. In our faultless protocol, $r$ receives $T$ private inputs $M_1, M_2, \ldots, M_T$ each of $\Theta(\log n)$ bits. $r$ takes $T$ rounds to broadcast each input, broadcasting $M_i$ at round $i$.

Now consider a simulation of our protocol and assume for the sake of contradiction that it succeeds with constant probability. Let $C = C(p)$ be a constant such that $1 - p \ge \exp(- C)$ where $p = \Omega(1)$ is the fault probability of our noisy network. By the non-coding assumption, all messages sent by $P'$ must be in the set $\{M_1, \ldots, M_T\}$. Denote by $t_i$ the number of times $r$ broadcasts $M_i$. For the sake of contradiction, assume that $\sum_{i=1}^T t_i \le \frac{1}{100 C} T \log \Delta$. Hence $\min_{i \in [T]} t_i \le \frac{1}{100 C} \log \Delta$ by an averaging argument. Let $i^* = \arg\min_{i \in [T]} t_i$. Notice that the probability that a fixed node receives a message is independent of that of any other node since WLOG only $r$ ever broadcasts. Therefore, the probability that a fixed node does not receive $M_{i^*}$ is $(1-p)^{t_{i^*}} \ge (1-p)^{\frac{1}{100 C} \log \Delta} \ge \exp(- C \frac{1}{100 C} \log \Delta) \ge \Delta^{- 1/100}$ by definition of $C$. Consequently, the probability that some node does not receive $i^*$ is at least $1 - (1 - \Delta^{-1/100})^\Delta \ge 1 - \exp(\Delta^{99/100}) $ which tends to $1$ as $\Delta \rightarrow \infty$. This contradicts our assumption that our simulation succeeds with constant probability. Therefore, no simulation protocol of length $\Omega(T \log \Delta)$ can deliver all messages with constant probability.
\end{proof}

\subsection{Simulations that Respect Information Flow} 
In this section we show how any simulation with less than a $\Omega(\poly(\log \Delta ))$ multiplicative overhead must route information along different paths than those in the faultless setting. In particular, we show that a $\Omega(\poly(\log \Delta ))$ lower bound holds in a directed network where information in any simulation flows just as it does in the noiseless setting.


\begin{restatable}{lemma}{Ldirected}
  \label{L:directed}
  Let $G = (L \cup R, E)$ be a complete directed bipartite graph with $|L| = |R| = \Delta$ that has an arc $(l, r)$ for all $l \in L, r \in R$. There exists a protocol $P$ of length $\Delta$ for the faultless setting on the directed network $G$ such that any protocol that works in the noisy setting with constant success probability requires $\Omega(\Delta \log \Delta)$ rounds. This bound holds even in the global control setting.
\end{restatable}
\gdef\ProofLdirected{
\begin{proof}
  Let $L = \{ l_1, l_2, \ldots, l_\Delta \}$ and $R = \{ r_1, r_2, \ldots, r_\Delta \}$ be the set of nodes on both sides of the partition. Every node $l_i \in L$ gets private input $M_i$ and needs to broadcast it to all nodes in $R$. In the faultless protocol $P$, $l_i$ broadcasts $M_i$ in round $i \in [\Delta]$.

  Now consider a noisy protocol $P'$. Since $G$ is directed, the only node from $L$ that has knowledge of $M_i$ is $l_i$. Moreover, we can assume without loss of generality that in any one round of $P'$ at most one node in $L$ broadcasts, since if this were not the case either no node in $R$ would be sent a message or a collision would occur at every node in $R$. Let $t_i$ be the number of rounds in $P'$ that $l_i$ broadcasts $M_i$. For the sake of contradiction, assume that the length of $P'$ is $c \Delta \log \Delta$ for a sufficiently small constant $c > 0$. Therefore, $\sum_{i=1}^\Delta t_i \le c \Delta \log \Delta$ and there exists $i^* \in [\Delta]$ such that $t_{i^*} \le c \log \Delta$. Since $M_{i^*}$ can only be broadcasted from $l_{i^*}$, $M_{i^*}$ is broadcasted at most $c \log \Delta$ times by an averaging argument. Thus, we have a star with central node $l_{i^*}$ and leaves given by $R$ with the assumption that  $c \log \Delta$ rounds suffices to spread a message from $l_{i^*}$ to every node in $R$. The remainder of the proof is identical to the strategy given in \Cref{L:non-coding} and hence is omitted.
\end{proof}
}\fullOnly{\ProofLdirected}

\subsection{Unconditional Lower Bound Hypothesis}
We do not believe there exists an $o(\poly(\log \Delta))$ multiplicative overhead simulation, even for non-adaptive protocols, and in this section we put forward a candidate hard example that might be used to prove this claim.

\textbf{Construction.} Let $G$ be a bipartite network with partition $(L, R)$, where $|L| = |R| = n$. Divide the nodes in $L$ into $\Delta$ groups of size $\frac{n}{\Delta}$, namely $L^1, \ldots, L^\Delta$. Let $l^i_1, \ldots, l^i_{n/\Delta}$ be the nodes in $L^i$. We repeat the following for $t = 1, 2, \ldots, \Delta$ iterations: pick a fresh independent permutation $\pi : [n] \to [R]$ and divide $R$ into $\Delta$ groups of size $n / \Delta$ according to $\pi$. Specifically, let $R^1 = \{ \pi(1), \pi(2), \ldots, \pi(n/\Delta) \}, R^2 = \{ \pi(n/\Delta + 1), \ldots, \pi(2 n/\Delta) \}$, ..., $R^\Delta = \{ \pi(n - n/\Delta + 1), \ldots, \pi(n) \}$. Note that the grouping $R$ changes between iterations unlike $L$ which remains fixed. Fully connect $l_t^i$ to all the nodes in $R^i$ for all $i \in [\Delta]$.
 
\begin{figure}
\centering
\includegraphics[scale=.3]{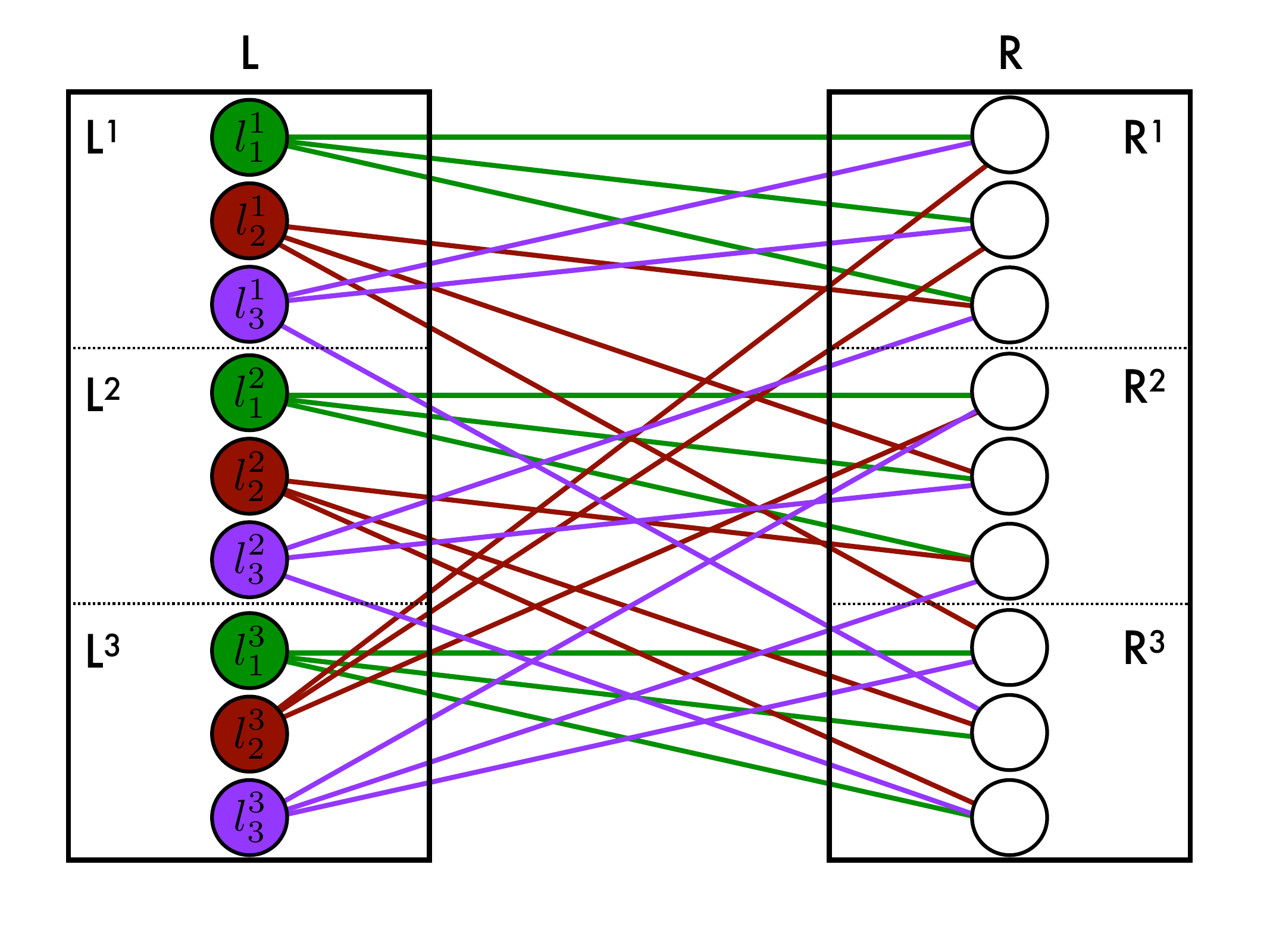}
\caption{A particular sample of our construction which we believe could help prove an unconditional lower bound. $R_i$ labeled according to the first iteration of our construction. $\Delta = 3$, $n = 9$. Nodes (and incident edges) colored according to the round in which they broadcast in the faultless protocol.}
\label{fig:conjLB}
\end{figure}

\textbf{Why we believe this protocol is hard.} Directly forwarding messages from a node $l \in L$ to each one of $l$'s neighbors requires a multiplicative $\Omega(\log \Delta)$ overhead before all of the neighbors receive the message. Thus, if we assume there is a $o(\log \Delta)$ overhead protocol, there must be a large fraction of messages that are delivered indirectly. That is, many nodes in $R$ receive many of the messages they need to simulate the original protocol from a different nodes than they do in the faultless protocol. However, indirectly delivering messages seems to require strictly more rounds than directly sending messages. Each $r \in R$ roughly wants to receive private input from a random subset of $L$. Therefore, if $r \in R$ receives a message indirectly from a neighbor $l \in L$, it is unlikely that the neighbors of $l$ apart from $r$ need this message to simulate the original protocol. Thus, while $l$ delivers a message to $r$, it blocks all other neighbors of $l$ from receiving messages they need to simulate the protocol. Lastly, we note that this problem roughly corresponds to a random instance of \emph{index coding}~\cite{bar2011index}, for which the bounds are currently not fully understood.

\section{Conclusion}\label{sec:conc}
In this paper we compared the computational power of noisy radio networks relative to classic radio networks. We demonstrated that any non-adaptive radio network protocol can be simulated with a multiplicative cost of $\poly(\log \Delta, \log \log n)$ rounds. Moreover, we demonstrated that a general radio network protocol can be simulated with a multiplicative $O(\Delta \log ^2 \Delta)$ round-overhead. 
We also showed lower bounds that suggest that any simulation of a radio network protocol by a noisy radio network requires $\Omega(\log \Delta)$ additional rounds multiplicatively.

Though we give a simulation for general protocols, the main focus of technical content of this work has been non-adaptive protocols. Non-adaptive protocols capture a number of well-known protocols (e.g.\ the optimal broadcast algorithm of \citet{gkasieniec2007faster}) and this focus has enabled us to understand both how delays propagate (see our blaming chain arguments) and how to efficiently exchange information among nodes (see our distributed binary search) in noisy radio networks. Having given solutions to these challenges, we now state promising directions for future work in noisy radio networks.

Recall that the third challenge for local-synchronization-based simulations described in \Cref{sec:intro} is that nodes cannot distinguish between not receiving a message in the original protocol and a message being dropped by a random fault. Our general protocol solves this issue but only through a costly subroutine in which every node shares information with all of its neighbors, thereby incurring an $O(\Delta)$ overhead. We leave as an open question whether this challenge can be overcome with $\poly (\log \Delta)$ overhead. ``Backtracking'' on faulty progress has been the subject of some interactive coding literature---see \citet{haeupler2014interactive}---and will likely prove insightful on this front. 


As a final direction for future work we note that many of our techniques apply to simulations of faulty versions of other models of distributed computing. For instance, applying the techniques of our general protocol simulation to a receiver fault version of CONGEST almost immediately yields a simulation of CONGEST with receiver faults by CONGEST with a $O(\log \Delta)$ multiplicative round overhead. We expect further applications abound.

\newpage

\bibliographystyle{plainnat}
\bibliography{main}

\newpage
\appendix
  
\shortOnly{
\section{Deferred Proofs}\label{sec:deferred-proofs}

\LmainRoutine*
\ProofLmainRoutine
\LlearnDelays*
\ProofLlearnDelays
\LdistToActive*
\ProofLdistToActive
\Lbroadcast*
\ProofLbroadcast
\staticSimulationTheorem*
\ProofStaticSimulationTheorem
\Ldirected*
\ProofLdirected
}

\section{Tail Bounds}\label{sec:tail-bounds}
In this section we prove the tail bounds that we use throughout this paper.
\begin{lemma}
  \label{lemm:tail-bounds-for-tail-bounds}
  Let $X_1, \ldots, X_n$ be independent random variables taking on integer values. Let $\mu_i = \E[X_i]$ and suppose that for every $i$ and every $t > 0$ we have $\Pr[X_i - \mu_i \ge t] \le \exp(-ct)$ for some $c > 0$. Then $\Pr[\sum_{i \in [n]} X_i - \mu_i \ge t] \le \exp(-\frac{c}{3} t)$ for $t \ge L\cdot n$, where $L > 0$ is a constant that depends only on $c$.
\end{lemma}
\begin{proof}
  We use $V^+ = \max(V, 0)$ to denote the positive part of a real number; note that $\Pr[(X_i - \mu_i)^+ \ge t] \le \exp(-c t)$ since probability is bounded above by $1$. Fix $t \in \mathbb{Z}_{\ge 0}$; any event where $\sum_{i \in [n]} X_i - \mu_i \ge t$ gives a non-negative integer sequence $u_i = (X_i - \mu_i)^+$ where $\sum_{i \in [n]} u_i \ge t$. This, in turn, gives a non-negative integer sequence $u_i \le (X_i - \mu_i)^+$ where $\sum_{i \in [n]} u_i = t$, for instance, by arbitrary decreasing positive $u_i$ until the sum is $t$. We will bound the probability of such a sequence existing.

Fix any non-negative integer sequence $u_i \in \mathbb{Z}_{\ge 0}$ where $\sum_{i \in [n]} u_i = t$ and $u_i \le (X_i - \mu_i)^+$. It is well-known combinatorial fact that there are $\binom{t+n-1}{n-1} \le \left (\frac{e(t + n - 1)}{n-1}\right)^{n-1} \le e^n (1 + \frac{t}{n})^n$ such sequences. Fixing one such sequence $u_i$, we have
  \begin{align*}
    & \Pr[\forall i \in [n], \; u_i \le (X_i - \mu_i)^+] \\
    \le & \prod_{i \in [n]} \Pr[(X_i - \mu)^+ \ge u_i] \\
    \le & \exp\left(- c \sum_{i \in [n]} u_i \right) = \exp(- c t) .
  \end{align*}
  
  Finally, union bounding over at most $e^n (1 + \frac{t}{n})^n$ sequences, we get $\Pr[\sum_{i \in [n]} X_i - \mu_i \ge t] \le e^n (1 + \frac{t}{n})^n \exp(-c t)$. We can choose a sufficiently large $L$ such that both $1 + x \le \exp(\frac{c}{3} x)$ for $x \ge L$, and $e^{1/L} \le \exp(\frac{c}{3} x)$ hold. In this case $\Pr[\sum_{i \in [n]} X_i - \mu_i \ge t] \le \exp(\frac{t}{L}) \exp(\frac{c}{3} t) \exp(-c t) \le \exp(-c/3 \cdot t)$.
\end{proof}

\begin{lemma}
  \label{lem:chernoff-sum-of-max}
  Let $Y_1, Y_2, \ldots, Y_T$ be independent maximums of at most $\Delta \ge 2$ independent geometric random variables with constant expectation. Then $\Pr[\sum_{i \in [T]} Y_i \ge C ( T \log \Delta + t)] \le \exp(-t)$ for all $t \ge 0$, where $C > 0$ is a constant.
\end{lemma}
\begin{proof}
  Geometric random variable $G$ with constant expectation have a natural tail bound of the form $\Pr[G \ge t] \le \exp(-c t)$. Since there are only finitely many geometric random variables in question, we can choose a sufficiently small constant $c > 0$ such that the bound holds for all of them. Consequently, the maximum of at most $\Delta$ such variables exhibits a tail bound. Let $G_1, G_2, \ldots, G_k$ where $k \le \Delta$ be geometric random variables with constant expectation and set $Y := \max(G_1, \ldots, G_k)$. Then $\Pr[Y \ge t] \le \sum_{i \in [k]} \Pr[G_i \ge t] \le \Delta \exp(-c t)$. Setting $\mu := \frac{\ln \Delta}{c}$ and $t \gets t' + \mu$, we get that $\Pr[Y-\mu \ge t'] \le \exp(-c t')$.

  Having obtained a tail bound for $Y$'s, we can apply \Cref{lemm:tail-bounds-for-tail-bounds}, getting
  \begin{align*}
    \Pr\left[\left(\sum_{i \in [T]} Y_i\right) - T \frac{\log \Delta}{c} \ge t\right] \le \exp\left(-\frac{c}{3} t \right) && \text{when $t \ge L \cdot T$} .
  \end{align*}
  By setting $t \gets L \cdot T + t'$ and $c' := \frac{1}{c} + L$, we get that $\Pr\left[\left(\sum_{i \in [T]} Y_i\right) \ge c' \cdot T \log \Delta + t'\right] \le \exp(-\frac{c}{3} t')$ for all $t' \ge 0$. Finally, setting $t' \gets \frac{3}{c} t''$ and $C := \max(c'', \frac{3}{c})$ we get the final result $\Pr\left[\left(\sum_{i \in [T]} Y_i\right) \ge C ( T \log \Delta + t'') \right] \le \exp(- t'')$ for all $t'' \ge 0$.
  
\end{proof}

\end{document}